\documentclass{article}

\usepackage[margin=1.3in]{geometry}
\usepackage{amsmath,amsfonts,graphicx,amssymb}
\usepackage{amsthm}
\usepackage{ifthen}
\usepackage{color}
\usepackage{caption}
\usepackage{subcaption}
\usepackage{mathtools}
\usepackage{bbm}
\usepackage{multirow}
\usepackage{cite}
\usepackage{footmisc}

\newtheorem{corollary}{Corollary}
\newtheorem{lemma}{Lemma}

\newtheorem{theorem}{Theorem}

\newtheorem{definition}{Definition}

\makeatother

\newcommand{\ra}{\rightarrow} 
\newcommand{\ua}{\uparrow}
\newcommand{\da}{\downarrow}
\newcommand{\prob}[1]{P\left(#1\right)} 

\newcommand{\R}{\mathbb{R}} 

\newcommand{\Exp}[1]{\mathbb{E}\left[#1\right]} 

\newcommand{\indicator}{\boldsymbol{1}}
\newcommand{\ceil}[1]{\left\lceil #1 \right\rceil}
\newcommand{\floor}[1]{\left\lfloor #1 \right\rfloor}
\DeclareMathOperator*{\argmax}{arg\,max}
\DeclareMathOperator*{\argmin}{arg\,min}
\newcommand{\ignore}[1]{}

\newcommand{\fr}[1]{\left\{#1\right\}}
 
\newcommand{\pd}[2]{\frac{\partial #1}{\partial #2}}
\newcommand{\al} [1] {\alpha_{#1} } 
\newcommand{\bt} [1] {\beta_{#1} } 
\newcommand{\g} [1] {\gamma_{#1} }

\newcommand{\bP}{\mathbb{P}}

\newcommand{\bE}{\mathbb{E}}

\newcommand{\bo}{B_{ov}}

\newcommand{\bn}{\mathbf{n}}
\newcommand{\bm}{\mathbf{m}}

\newcommand{\sN}{\mathcal{N}}
\newcommand{\sNo}{\mathcal{N}^\circ}

\newboolean{showcomments}
\setboolean{showcomments}{true}
\newcommand{\an}[1]{  \ifthenelse{\boolean{showcomments}}
{ \textcolor{red}{(AN says:  #1)}} {}  }
\newcommand{\jk}[1]{  \ifthenelse{\boolean{showcomments}}
{ \textcolor{blue}{(JK says:  #1)}} {}  }
\newcommand{\dm}[1]{  \ifthenelse{\boolean{showcomments}}
{ \textcolor{green}{(DM says:  #1)}} {}  }
\newcommand{\bp}[1]{  \ifthenelse{\boolean{showcomments}} 
{ \textcolor{red}{(Bala says:  #1)}} {}  }

\title{Sharing within limits: Partial resource pooling in loss systems\thanks{A preliminary version of this work appeared in the proceedings of COMSNETS 2016 \cite{Kumar16}.}}

\author{Anvitha~Nandigam\thanks{Anvitha Nandigam, D.~Manjunath, and Jayakrishnan Nair are with the Department of Electrical Engineering, IIT Bombay.} 
	\and Suraj~Jog\thanks{Suraj Jog is with the EECS department at the University of Illinois Urbana-Champaign.}
	\and D.~Manjunath\footnotemark[2] \and Jayakrishnan~Nair\footnotemark[2] 
	\and B.~J.~Prabhu\thanks{B.~J.~Prabhu is with LAAS-CNRS, Universit\'e de Toulouse, CNRS, Toulouse, France.} 
}

\begin{document}


\maketitle

\begin{abstract}
Fragmentation of expensive resources, e.g., spectrum for wireless
services, between providers can introduce inefficiencies in resource
utilisation and worsen overall system performance. In such cases,
resource pooling between independent service providers can be used to
improve \ignore{overall system} performance.
However, for providers to agree to pool their resources, the
arrangement has to be mutually beneficial. The traditional notion of
resource pooling, which implies complete sharing, need not have this
property. For example, under full pooling, one of the providers may be
worse off and hence have no incentive to participate. In this paper,
we propose \emph{partial} resource sharing models as a generalization
of full pooling, which can be configured to be beneficial to all
participants.

We formally define and analyze two partial sharing models between two
service providers, each of which is an Erlang-$B$ loss system with the
blocking probabilities as the performance measure. We show that there
always exist partial sharing configurations that are beneficial to
both providers, irrespective of the load and the number of circuits of
each of the providers. A key result is that the Pareto frontier has at
least one of the providers sharing all its resources with the
other. Furthermore, full pooling may not lie inside this Pareto
set. The choice of the sharing configurations within the Pareto set is
formalized based on bargaining theory. 
Finally, large system approximations of the blocking probabilities in
the quality-efficiency-driven regime are presented.

\end{abstract}

\section{Introduction}
\label{sec:intro}

High availability is an important requirement of many services like
wireless communications, cloud computing, hospitals, and fire fighting
services. The resources required to provide these are expensive ---
think spectrum and base stations for wireless communication, servers
and associated infrastructure for cloud computing, medical equipment
and doctors for hospitals, fire trucks and trained personnel for fire
fighting services. Service denial, which is the inability of the
resources to satisfactorily meet a fraction of the demand, is an
important performance measure for these services. When the demand is
stochastic, the amount of resources required to provide a prescribed
grade of service may be such that the utilization is low, especially
in smaller systems. This means that small providers require more
resources for a given service level. This in turn can make these
services expensive for small providers. However, large systems
experience statistical multiplexing gains and hence achieve economies
of scale. Thus resource sharing or resource pooling can be useful when
there are several independent entities providing similar services
using similar resources.

Typically, resource pooling is assumed to involve the combining of the
resources of all the participating providers and treating the combined
system as one unit. In this paper we propose \emph{partial resource
  pooling} as a generalization of the full pooling
models. Specifically, we consider two loss systems modeled as M/M/N/N
queues that operate independently in that they manage their own calls
but they cooperate by pooling their servers partially as follows. When
an overflow call arrives at one of the systems, (i.e., the number of
active calls of the provider is greater than the number of servers it
has), the other provider \emph{may} loan one of its free servers in
which case the call will be admitted. The server is loaned for the
duration of the call. The overflow call is lost if the other provider
chooses not to loan the server. The partial sharing model determines
when such an overflow call is admitted. At one extreme would be the no
pooling case where all overflow calls are lost and at the other
extreme is the full pooling case where all overflow calls are admitted
if there is a free server.

As mentioned above, several resource pooling models are available in
the literature with the key feature being independent service systems,
managed by independent decision makers, cooperating \emph{fully,}
acting as a single entity, and sharing the costs and/or benefits
suitably. In other words, it is an all-or-nothing game with the
parties either pooling their resources completely or staying out of the
coalition and operating on their own. These models typically use cooperative
or coalitional game theoretic ideas to determine the answers to the
following questions.~(1)~Which entities will form a cooperating
unit?~(2)~How are the revenues and costs shared?

In \cite{Sarkar08,Sarkar09} independent wireless network operators
share base station infrastructure and spectrum to efficiently serve
their customers. Stable cost sharing arrangements between the network
operators are explored in this setting. Note that the sharing model
here involves complete pooling of the spectrum and the base stations,
as opposed to the opportunistic sharing of resources with secondary
users in cognitive radio systems (e.g., \cite{Aykildiz06,Zhao07}). In
the system studied in \cite{Karsten15}, the cooperating entities
choose the quantity of resources to provide a specified service grade
and stable cost sharing arrangements are determined.

Server pooling has also been studied in the context of reengineering
of manufacturing lines by modeling them as Jackson networks. Here
several nodes (service stations) are combined into one service station
that is capable of providing the services of all the components, e.g.,
\cite{Mandelbaum98} and references therein.

More abstract forms of resource pooling have also been considered in
the queueing literature. In \cite{Anily10}, cooperating single server
queues are combined into one single server queue whose service rate is
upper bounded by the sum of the capacities. The actual service rate is
determined by a cost structure and the service grade. In \cite{Ozen11},
cooperation among queues to optimally invest in a common service capacity,
or choose the optimal demand to serve as a common entity, is analyzed.

To motivate our break from the preceding literature, consider the
following example of two M/M/N/N loss systems, e.g., cellular service
providers with a fixed number of channels. Provider $P_1$, with $85$
channels and a load of 88 Erlangs, has a blocking probability of $0.1$
and Provider $P_2$, with $59$ channels and $70$ Erlangs load, has a
blocking probability of $0.2$. If the two providers are combined into
one, the joint system would have a combined load of $158$ Erlangs
served by $144$ channels with blocking probability of $0.11$. Clearly
cooperation is beneficial to $P_2,$ but unacceptable to $P_1.$ And if
blocking probability were the only performance measure, it is a case
of ``and never the twain shall meet.'' The partial pooling mechanisms
that we develop in this paper allow \emph{both} the operators to
improve their performance.



The rest of the paper is organized as follows. In the next section, we
introduce the system model and describe two partial sharing models:
the \emph{bounded overflow} sharing model and the \emph{probabilistic}
sharing model. The blocking probabilities under these models and their
monotonicity properties are also derived. In Section~\ref{sec:pareto},
we characterize
the \emph{Pareto frontier} of the sharing configurations. The key
result is that the Pareto frontier is non empty and is at the boundary
of all possible sharing configurations---one of the providers has to
always yield its free servers to overflow calls of the other. In
Section~\ref{sec:economics}, we characterize the economics of partial
sharing by treating the sharing that emerges as the solution of Nash
bargaining, Kalai-Smorodinsky bargaining, egalitarian sharing (both
parties experience the same benefit) and utilitarian sharing (maximize
the system benefit).  The utility sets over which these bargaining
solutions will be computed for our model do not satisfy the usual
properties of convexity or comprehensiveness, making it less
straightforward to guarantee the uniqueness of the bargaining
solution. Nevertheless, using monotonicity properties of the blocking
probabilities shown in Section~\ref{sec:pareto}, we are able to
show uniqueness of the Kalai-Smorodinsky and egalitarian
solutions. Via numerical experiments, we demonstrate the contrasts
between the different bargaining solutions, and also the potential
benefits of partial resource pooling for both providers.
In Section~\ref{sec:large}, we address the computational complexity of
the blocking probabilities for large loss systems \cite{Kelly91}. We
consider large system limits under the well known
quality-efficiency-driven (QED) regime. Our large system analysis
provides computationally light, yet accurate approximations of the
blocking probabilities for realistic system settings.
Finally, we conclude with a discussion on alternate sharing models,
connections to more familiar models from the circuit multiplexing
literature, alternate applications, and future work in
Section~\ref{sec:discuss}.



\section{Model and Preliminaries}
\label{sec:model}

In this section, we describe our system model, propose our mechanisms
for partial resource pooling, and state some preliminary results.

We begin by describing the baseline model with no resource pooling.
We consider two service providers, $P_1$ and $P_2.$
Each provider is modeled as an M/M/N/N queue or an Erlang-B loss
system. Specifically, $P_i$ has $N_i$ servers/circuits. Calls arrive
for service at $P_i$ according to a Poisson process of rate
$\lambda_i.$ When a call arrives, it begins service at a free server
if one is available. If all servers are busy, then the call is
blocked. The holding times (a.k.a. service times) of calls at $P_i$
are i.i.d., with $S_i$ denoting a generic call holding time. We assume
that $\Exp{S_i} =: \frac{1}{\mu_i} < \infty.$ Thus, the offered load
seen by $P_i$ is given by $a_i := \frac{\lambda_i}{\mu_i}.$ With no
resource pooling between the providers, it is well known that the
steady state call blocking probability for $P_i$ is given by the
Erlang-B formula:
\begin{equation*}
  \label{eq:erlangB}
  E(N_i,a_i) = \dfrac{a_i^{N_i}}{N_i!}\left[ 
      {\sum_{j = 0}^{N_i} \dfrac{a_i^{j}}{j!}} \right]^{-1}.
\end{equation*}
It is also well known that the steady state call blocking probability
is insensitive to the \emph{distribution} of the call holding times,
i.e., it depends only on the \emph{average} call holding
time. Moreover, the blocking probability depends on the workload only
through the offered load~$a_i.$

Next, we describe the proposed partial resource pooling models.

\subsection{Probabilistic sharing model}

The probabilistic sharing model is parameterized by the tuple
$(x_1,x_2) \in [0,1]^2.$ Informally, under this model, $P_i$ accepts
an overflow call from $P_{-i}$ with probability $x_i.$\footnote{When
  referring to the provider labeled~i, we use $-i$ to refer to the
  other provider.\label{footnote:minusi}}

Formally, the probabilistic sharing model is defined as follows. Let
$n_i$ denote the number of active calls of $P_i.$ When a call of
$P_{-i}$ arrives,
\begin{itemize}
\item If $n_{-i} < N_{-i},$ and $n_1 + n_2 < N_1 + N_2,$ the call is
  admitted
\item If $n_{-i} \geq N_{-i}$ and $n_1 + n_2 < N_1 + N_2,$ the call is
  admitted with probability~$x_{i}$
\item If $n_1 + n_2 = N_1 + N_2,$ the call is blocked
\end{itemize}
The vector $x := (x_1,x_2)$ defines the (partial) sharing
configuration. Note that $x_i$ captures the extent to which $P_i$
pools its resources with $P_{-i}.$ In particular, the configuration
$(0,0)$ corresponds to no pooling, and the configuration $(1,1)$
corresponds to complete pooling. Moreover, note that the probabilistic
sharing model does not keep track of whether an ongoing call of $P_i$
is occupying a server of $P_i$ or $P_{-i}.$ This simplification, which
makes the model analytically tractable, is identical to the
\emph{maximum packing} or \emph{call repacking} model of
\cite{Everitt83,Kelly85} and has been used extensively in the
literature. One interpretation of this assumption is that once a $P_i$
server becomes free, if there are any ongoing $P_i$ calls on $P_{-i}$
servers, one of those is instantaneously shifted to the free $P_i$
server.

Next, we characterize the steady state blocking probabilities under
this partial sharing model. To do so, we define the following subsets
of $\mathbb{Z}_+^2.$
\begin{align*}
  M^{(p)} &:= \left\{(n_1,n_2):\ n_1 + n_2 \leq N_1 + N_2 \right\} \\
  R^{(p)} &:= \left\{(n_1,n_2):\ n_1 + n_2 = N_1 + N_2 \right\}            
\end{align*}
For $i \in \{1,2\},$
\begin{align*}
  D_i^{(p)} &:= \left\{(n_1,n_2):\ n_i \geq N_i,\ n_1 + n_2 < N_1 + N_2 \right \}.        
\end{align*}
Here $M^{(p)}$ refers to the set of feasible states, $R^{(p)}$
corresponds to the feasible states when all the servers are busy, and
$D_i^{(p)}$ are the states in which calls of $P_i$ are accepted with
probability~$x_{-i}.$

\begin{lemma}
  \label{lemma:prob_sharing}
  Under the probabilistic sharing model, the steady state blocking
  probability for Provider~$i$ is given by 
  \begin{align*}
    B_i^{(p)}(x_1,x_2) =& \frac{1}{G}\biggl[ \sum_{(n_1,n_2) \in R^{(p)}} f_1(n_1)f_2(n_2) \\
                        &\quad + \sum_{(n_1,n_2) \in D_i^{(p)}}
                          f_1(n_1)f_2(n_2)(1-x_{-i}) \biggr],
  \end{align*}
  where 
\begin{align*}
  f_i(n)  & =
     \begin{cases}
       a_i^{n} / n! & \mbox{if $n < N_i$} \\
       a_i^{n}x_{-i}^{n-N_i} / n! & \mbox{if $N_i \leq n \leq N_1 + N_2$}
     \end{cases},
                                    \nonumber \\
  G & = \sum_{(n_1,n_2) \in M^{(p)}} f_1(n_1) f_2(n_2).
\end{align*} 
\end{lemma}
A key takeaway from Lemma~\ref{lemma:prob_sharing} is that under the
probabilistic sharing model, the steady state blocking probabilities
remain insensitive to the distributions of the call holding
times. Moreover, the dependence of the incoming workload on each
provider's blocking probability is only through the vector of offered
loads $(a_1,a_2).$ Finally, note that
\begin{align*}
  &B_i^{(p)}(0,0) = E(N_i,a_i),\\
  &B_1^{(p)}(N_1,N_2) = B_2^{(p)}(N_1,N_2) = E(N_1 + N_2, a_1 + a_2).           
\end{align*}
\begin{proof}
  Assuming that the call holding times are exponentially distributed,
  the state $(n_1,n_2)$ of the system evolves as a continuous time
  Markov chain (CTMC) over $M.$ It is easy to check that this CTMC is
  time-reversible and its invariant distribution $\pi$ has a product
  form: $$\pi(n_1,n_2) = \frac{f_1(n_1) f_2(n_2)}{G.}$$ The steady
  state blocking probability is then obtained by invoking the PASTA
  property.

  The insensitivity of the blocking probabilities to the call holding
  time distributions is a direct consequence of the reversibility of
  the above CTMC~\cite{Schassberger86}.
\end{proof}

  
\subsection{Bounded overflow pooling model}

The bounded overflow (BO) model is parameterized by the tuple
$(k_1,k_2),$ where $k_i \in [0,N_i].$ Informally, under the BO model,
$P_i$ accepts up to $k_i$ overflow calls from the other provider
$P_{-i}.$ Thus, $k_i$ is indicative of the extent to which $P_i$
shares its resources with $P_{-i}.$ We use randomization to let $k_i$
take real values in $[0,N_i];$ specifically, $P_i$ admits up to
$\floor{k_i}$ overflow calls from $P_{-i}$, and admits a
$\ceil{k_i}$-th overflow call with probability $\fr{k_i},$ where
$\fr{k_i} := k_i - \floor{k_i}$ denotes the fractional part of $k_i.$

Formally, the BO model is defined as follows. Recall that $n_i$
denotes the number of active calls of $P_i.$
When a call of $P_{-i}$ arrives,
\begin{itemize}
\item If $n_{-i} < N_{-i} + \floor{k_{i}}$ and $n_1 + n_2 < N_1 +
  N_2,$ the call is admitted
\item If $n_{-i} = N_{-i} + \floor{k_{i}}$ and $n_1 + n_2 < N_1 +
  N_2,$ the call is admitted with probability $\fr{k_{i}}$
\item Else, the call is blocked
\end{itemize}

We refer to the tuple $(k_1,k_2)$ as the (partial) sharing
configuration between $P_1$ and $P_2.$ Under the BO model, $P_i$ can
have at most $N_i + \ceil{k_{-i}}$ concurrent calls. Note that $(0,0)$
corresponds to no resource pooling, and $(N_1,N_2)$ corresponds to
full pooling between the providers. Finally, we note that the BO model
also assumes call repacking \cite{Everitt83,Kelly85}.

Next, we characterize the blocking probability of each provider under
the BO model. To express the blocking probabilities, we define the
following subsets of $\mathbb{Z}_+^2.$
\begin{align*}
  M^{(bo)} &:= \left\{(n_1,n_2):\  \substack{n_1 \leq N_1 + \ceil{k_2}
  \\n_2 \leq N_2 + \ceil{k_1} \\n_1 + n_2 \leq N_1 + N_2} \right\} \\
  R^{(bo)} &:= \left\{(n_1,n_2):\  \substack{n_1 \leq N_1 + \ceil{k_2}
  \\n_2 \leq N_2 + \ceil{k_1} \\n_1 + n_2 = N_1 + N_2} \right\}
\end{align*}
For $i \in \{1,2\},$
\begin{align*}
  C_i^{(bo)} &:= \left\{(n_1,n_2):\ n_i = N_i + \ceil{k_{-i}},\ n_{-i} < N_{-i} - \ceil{k_{-i}} \right \}, \\
  D_i^{(bo)} &:= \left\{(n_1,n_2):\ n_i = N_i + \floor{k_{-i}},\ n_{-i} < N_{-i} - \floor{k_{-i}} \right \}.        
\end{align*}
Here $M^{(bo)}$ refers to the set of feasible states, $R^{(bo)}$
corresponds to the feasible states when all the servers are busy,
$C_i^{(bo)}$ is the set of feasible states for which arriving calls of
$P_i$ are blocked due to the constraint on the number of overflow
calls, and $D_i^{(bo)}$ are the states for which calls of $P_i$ are
accepted with probability $\{k_{-i}\}.$

The following lemma characterizes the blocking probabilities of both
providers under the BO partial sharing model.

\begin{lemma}
  \label{lemma:bound_sharing}
  Under the bounded overflow sharing model, the steady state blocking
  probability for provider $P_i$ is given by
  \begin{align*}
    &B_i^{(bo)}(k_1,k_2) = \frac{1}{G}\biggl[ \sum_{(n_1,n_2) \in R^{(bo)} \cup C_i^{(bo)}} f_1(n_1) f_2(n_2) \\
    &\qquad + \indicator_{[\fr{k_{-i}} \neq 0]} (1-\fr{k_{-i}}) \sum_{(n_1,n_2) \in D_i^{(bo)}}  f_1(n_1) f_2(n_2)   \biggr],
  \end{align*}
  where
  \begin{align*}
    f_i(n) & = \left\{ \begin{array}{ll}
                         a_i^{n} / n!  &  \mbox{if $n \leq N_i + \floor{k_{-i}}$} \\
                         \fr{k_{-i}} a_i^{n} / n! & \mbox{if $n = N_i + \floor{k_{-i}} + 1$},
                       \end{array} \right.\\
    G & = \sum_{(n_1,n_2) \in M^{(bo)}} f_1(n_1) f_2(n_2).
\end{align*} 
\end{lemma}

As before, note that the steady state blocking probabilities are
insensitive to the distributions of the call holding times, and depend
only on the vector of offered loads $(a_1,a_2).$ Moreover,
\begin{align*}
  &B_i^{(bo)}(0,0) = E(N_i,a_i),\\
  &B_1^{(bo)}(N_1,N_2) = B_2^{(bo)}(N_1,N_2) = E(N_1 + N_2, a_1 + a_2).           
\end{align*}
The proof of Lemma~\ref{lemma:bound_sharing}, being similar to that of
Lemma~\ref{lemma:prob_sharing}, is omitted.

\subsection{Monotonicity properties of the blocking probabilities}

We conclude this section by collecting some monotonicity properties of
the blocking probabilities under the above partial sharing
models. These properties play a key role in our analysis of the game
theoretic aspects of partial sharing in Sections~\ref{sec:pareto}
and~\ref{sec:economics}.

When stating results that apply to both sharing models, we refer to
the steady state blocking probability of Provider~$i$ as
$B_i(x_1,x_2),$ with the understanding that this represents
\begin{itemize}
\item $B_i^{(p)}(x_1,x_2)$ under the probabilistic sharing model,
\item $B_i^{(bo)}(x_1 N_1, x_2 N_2)$ under the bounded overflow
  sharing model (i.e., $x_i = \frac{k_i}{N_i}$).
\end{itemize}
Note that the overall steady state blocking probability of the system
is given by
\begin{equation*}
  B_{ov}(x_1,x_2) = \frac{\lambda_1}{\lambda_1+\lambda_2} B_1(x_1,x_2)
  + \frac{\lambda_2}{\lambda_1+\lambda_2} B_2(x_1,x_2).
 \end{equation*}

Our monotonicity results are summarized in the following theorem.
\begin{theorem}
  \label{thm:monotonicity}
  Under the probabilistic as well as the bounded overflow partial
  sharing models, the steady state blocking probabilities satisfy the
  following properties, for $i \in (1,2).$
  \begin{enumerate}
  \item $B_i(x_1,x_2)$ is a strictly increasing function of $x_i.$
  \item $B_{-i}(x_1,x_2)$ is a strictly decreasing function of $x_{i}.$
  \item If $\mu_1 = \mu_2,$ then $B_{ov}(x_1,x_2)$ is a strictly
    decreasing function of $x_i.$
  \end{enumerate}
\end{theorem}

Theorem~\ref{thm:monotonicity} highlights the impact of an increase in
$x_i$ on the blocking probabilities of $P_i$ and $P_{-i},$ as well as
the overall blocking probability. In particular, an increase in $x_i$
(i.e., an increase in the extent to which $P_i$ shares its servers
with $P_{-i}$) decreases the fraction of blocked calls at $P_{-i},$ at
the expense of increasing the fraction of blocked calls at $P_i.$ Note
that Statements~1 and~2 imply that $(0,0)$ is the unique Nash
equilibrium between the providers, assuming that the utility of each
provider is a strictly decreasing function of its blocking
probability. This means that a non-cooperative interaction sans
signalling would not yield a mutually beneficial partial sharing
configuration between the providers. In contrast, we show in
Section~\ref{sec:economics} that a \emph{bargaining-based} interaction
would indeed result in mutually beneficial partial sharing
configurations.

Finally, Statement~3 of Theorem~\ref{thm:monotonicity} highlights that
so long as the mean call holding times are matched across both
providers, an increase in $x_i$ results in an overall reduction in the
call drop probability of the system. This is because increasing $x_i$
provides additional opportunities for calls to get admitted when there
are free circuits. In particular, Statement~3 above implies that for
$(x_1,x_2) \notin \{(0,0),(1,1)\},$
$$B_{ov}(1,1) < B_{ov}(x_1,x_2) < B_{ov}(0,0),$$ implying that
complete pooling minimizes the overall blocking probability of the
system (when $\mu_1 = \mu_2$).

Note that even through the statement of Theorem~\ref{thm:monotonicity}
applies compactly to both sharing models, a separate proof is required
for each model. We provide the proof of Theorem~\ref{thm:monotonicity}
for the bounded overflow sharing model in
Appendix~\ref{app:monotonicity_bo}, and for the probabilistic sharing
model in Appendix~\ref{app:monotonicity_prob}. It is important to
point out that while the statement of Theorem~\ref{thm:monotonicity}
seems intuitive, the proof is fairly non-trivial. In particular, our
proof of Statement~3 for the bounded overflow model involves a subtle
sample path argument (see Appendix~\ref{app:monotonicity_bo}).


\section{Efficient Partial Sharing Configurations}
\label{sec:pareto}

We have seen that complete resource pooling between providers is not
necessarily stable, in the sense that it is not guaranteed to be
beneficial to both providers. Having defined mechanisms for partial
resource sharing in Section~\ref{sec:model}, the natural questions
that arise are:
\begin{enumerate}
\item Do there exist stable partial sharing configurations?
\item If so, can one characterize the Pareto frontier of the space of
  partial sharing configurations?
\end{enumerate}
The goal of this section is to address the above questions.

First, we prove that under both the sharing mechanisms defined in
Section~\ref{sec:model}, there exist stable partial sharing
configurations, i.e., there exist partial sharing configurations that
result in a strictly lower blocking probability for each provider,
compared to the case of no pooling. Next, we focus on characterizing
the set of Pareto-efficient partial sharing
configurations. Intuitively, this is the set of `efficient' sharing
configurations, over which it is not possible to lower the blocking
probability for any provider without increasing the blocking
probability of the other. Our main result is that any Pareto sharing
configuration has at least one provider pooling all of its servers
(i.e., $x_i = 1$ for some $i$).\footnote{$P_i$ pooling all its servers
  means that it always yields a free server to an overflow call from
  $P_{-i}.$} Intuitively, efficient partial sharing configurations
involve the more congested provider pooling all of its servers,
enabling both providers to benefit from the resulting statistical
economies of scale. Finally, we provide an exact characterization of
the set of Pareto efficient sharing configurations (a.k.a. the Pareto
frontier) under the probabilistic and bounded overflow partial sharing
models. 

We begin by defining `stable' partial sharing configurations.
\begin{definition}
  A sharing configuration $(x_1,x_2)$ is \emph{QoS-stable} if
  $B_i(x_1,x_2) < E(N_i,a_i)$ for $i=1,2.$
\end{definition}
The following lemma guarantees the existence of QoS-stable sharing
configurations.
\begin{lemma}
  \label{lemma:qos_stable-non_empty}
  Under the probabilistic as well as the bounded overflow partial
  sharing models, the set of QoS-stable partial sharing configurations
  is non-empty.
\end{lemma}
Lemma~\ref{lemma:qos_stable-non_empty} essentially validates our
partial sharing mechanisms. Specifically, it asserts that even when the
providers are highly asymmetric with respect to capacity and/or
offered load, and even when complete resource pooling is not
beneficial to one of the providers, there exists a partial sharing
configuration that is beneficial to both providers. We omit the proof
of Lemma~\ref{lemma:qos_stable-non_empty} since it is a direct
consequence of Lemma~\ref{lemma:pareto_bo_improvement} below.

Now that we are certain that mutually beneficial partial sharing
configurations exist, we turn to the characterization of the set of
efficient configurations. We begin by defining Pareto-efficient
sharing configurations.

\begin{definition}
  A sharing configuration $(x_1,x_2)$ is \emph{Pareto-efficient}~if
  \begin{enumerate}
  \item $(x_1,x_2)$ is QoS-stable,
  \item there does not exist a sharing configuration $(x'_1,x'_2)$
    such that $B_i(x'_1,x'_2) \leq B_i(x_1,x_2)$ for all
    $i \in \{1,2\}$ and $B_i(x'_1,x'_2) < B_i(x_1,x_2)$ for some
    $i \in \{1,2\}.$
  \end{enumerate}
\end{definition}

Condition~(2) above is the standard definition of
Pareto-efficiency---a configuration is Pareto-efficient if it is not
possible to enhance the utility of one party (the utility of a
provider being a strictly decreasing function of its blocking
probability) without diminishing the utility of the other. Since our
interest is in capturing the set of configurations that the providers
could potentially agree upon, it is also natural to impose the
requirement that each provider stands to benefit from the partial
sharing agreement; this is captured by Condition~(1) in the
definition.

Our main result is that at any Pareto-efficient sharing configuration,
at least one provider pools all of its servers.

\begin{theorem}
  \label{thm:pareto}
  Under the probabilistic as well as the bounded overflow partial
  sharing models, the set of Pareto-efficient sharing configurations
  is non-empty. Moreover, any Pareto-stable sharing configuration
  $(x_1,x_2)$ satisfies the property that $x_i = 1$ for some
  $i \in \{1,2\}.$
\end{theorem}

Intuitively, if the providers are symmetric, full pooling ($x_i = 1$
for all $i)$ is Pareto-efficient, thanks to the statistical economies
of scale in the pooled system. Theorem~\ref{thm:pareto} highlights
that under general (possibly asymmetric) settings, where full pooling
may not be QoS-stable, efficient configurations still involve at least
one provider pooling all its servers. Indeed, statistical economies of
scale lie at the heart of this result as well, as is highlighted by
Lemma~\ref{lemma:pareto_bo_improvement} stated below, which forms the
basis of the proof of Theorem~\ref{thm:pareto}.

In stating Lemma~\ref{lemma:pareto_bo_improvement} and proving
Theorem~\ref{thm:pareto}, we use the following notation:
$\mathcal{X} := [0,1]^2,\quad \mathcal{X}^o := [0,1)^2.$

\begin{lemma}
  \label{lemma:pareto_bo_improvement}
  Under the probabilistic as well as the bounded overflow partial
  sharing models, for any $(x_1,x_2) \in \mathcal{X}^o,$ there exists
  $\theta > 0$ such that
  $$\nabla B_i(x_1,x_2)\cdot (1,\theta) < 0 \quad \forall i \in
  \{1,2\}.$$
\end{lemma}

Lemma~\ref{lemma:pareto_bo_improvement} implies that at any sharing
configuration $x \in \mathcal{X}^o$, it is possible to strictly
improve the blocking probability of both providers by increasing
\emph{both} components of $x$ (in the direction~$\theta$).\footnote{It
  is not hard to see that the blocking probabilities under the
  probabilistic sharing model (characterized in
  Lemma~\ref{lemma:prob_sharing}) are continuously differentiable over
  $\mathcal{X}.$ For the bounded overflow model, the blocking
  probabilities (characterized in Lemma~\ref{lemma:bound_sharing}) are
  continuous over $[0,N_1] \times [0,N_2]$ and differentiable for
  $k_1,k_2 \notin \mathbb{Z}_+.$ If $k_i$ is an integer, then the
  partial left and right derivatives with respect to $k_i$
  exist. Thus, for the bounded overflow model, the gradients in the
  statement of Lemma~\ref{lemma:pareto_bo_improvement} are understood
  to be composed of the right derivative with respect to $x_i$ when
  $x_i N_i$ is an integer.}

We now use Lemma~\ref{lemma:pareto_bo_improvement} to prove
Theorem~\ref{thm:pareto}.

\proof[Proof of Theorem~\ref{thm:pareto}] We provide a unified proof
of Theorem~\ref{thm:pareto} for both partial sharing models. Invoking
Lemma~\ref{lemma:pareto_bo_improvement} at the configuration $(0,0),$
we conclude that the set of QoS-stable configurations is non-empty.
  For $i\in \{1,2\},$ define
  $$\mathcal{B}_i(x_1,x_2) := \max(0, E(N_i,a_i) - B_i(x_1,x_2)).$$
  Consider the following optimization:
  $$\max_{x \in \mathcal{X}} \mathcal{B}_1(x_1,x_2)
  \mathcal{B}_2(x_1,x_2).$$ 
  Since this is the maximization of a continuous function over a
  compact domain, a maximizer $x^* \in \mathcal{X}$ exists. It is easy
  to see that $x^*$ is Pareto-efficient, implying that the set of
  Pareto-efficient configurations is non-empty. Finally,
  Lemma~\ref{lemma:pareto_bo_improvement} implies that no
  Pareto-stable configuration lies in $\mathcal{X}^o,$ implying that
  any Pareto-efficient configuration lies in
  $\mathcal{X} \setminus \mathcal{X}^o.$ This completes the proof.
  \endproof

\begin{figure*}[t!]
    \centering
    \begin{subfigure}[t]{0.3\textwidth}
        \centering
        \includegraphics[scale=0.4]{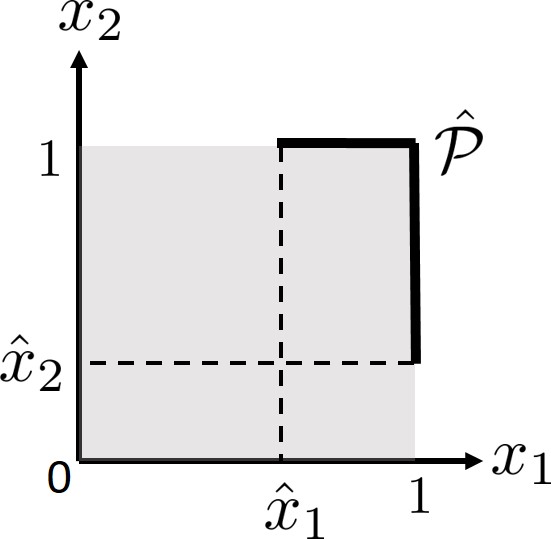}
        \caption{Case 1: Both providers strictly better off under full
          pooling}
 		 \label{fig:case1}
    \end{subfigure}%
    ~ 
    \begin{subfigure}[t]{0.3\textwidth}
        \centering
        \includegraphics[scale=0.4]{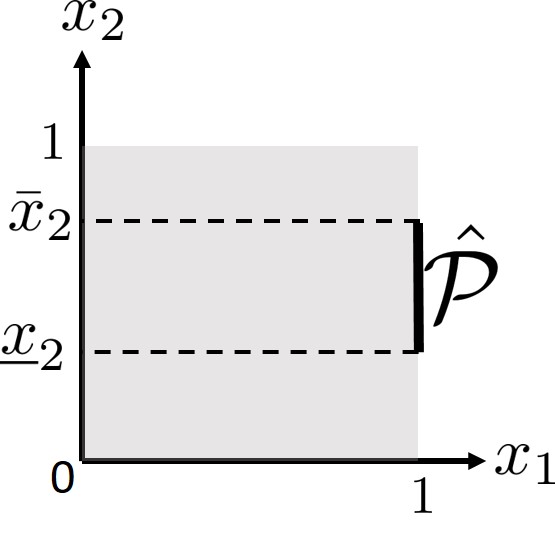}
        \caption{Case 2: Only Provider~1 strictly better off under
          full pooling}
  		\label{fig:case2}
    \end{subfigure}
    ~
    \begin{subfigure}[t]{0.3\textwidth}
        \centering
        \includegraphics[scale=0.4]{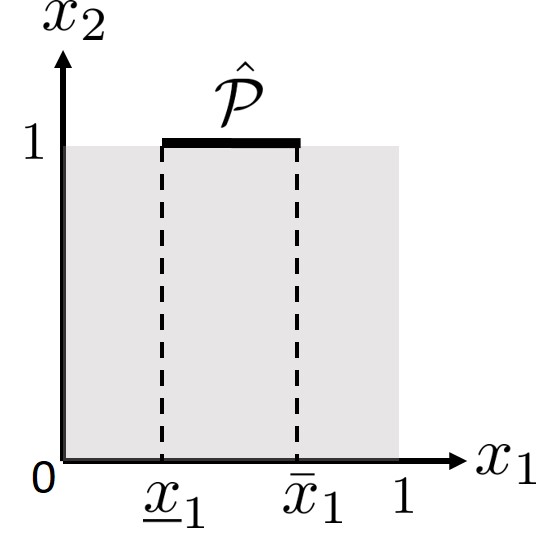}
        \caption{Case 3: Only Provider~2 strictly better off under
          full pooling}
  		\label{fig:case3}
    \end{subfigure}
    \caption{The set $\hat{\mathcal{P}}$ of Pareto-efficient partial
      sharing configurations}
    \label{fig:pareto}
\end{figure*}

It now remains to prove Lemma~\ref{lemma:pareto_bo_improvement}.
\proof[Proof of Lemma~\ref{lemma:pareto_bo_improvement}] We provide a
unified proof of Lemma~\ref{lemma:pareto_bo_improvement} for both
partial sharing models. $\nabla B_1(x_1,x_2)\cdot (1,\theta) < 0$ is
equivalent to
$$\theta > \frac{ \pd{B_1(x_1,x_2)}{x_1}}{ - \pd{B_1(x_1,x_2)}{x_2}}
=: \underline{\theta}.$$

Similarly, $\nabla B_2(x_1,x_2)\cdot (1,\theta) < 0$ is equivalent to
  $$\theta < \frac{- \pd{B_2(x_1,x_2)}{x_1}}{ \pd{B_2(x_1,x_2)}{x_2}}
  =: \bar{\theta}.$$
  
  We therefore have to prove that $\underline{\theta} < \bar{\theta},$
  which is equivalent to
  \begin{align}
    &\pd{B_1(x_1,x_2)}{x_1} \pd{B_2(x_1,x_2)}{x_2} \nonumber \\
    \label{eq:pareto_1}
    &\qquad \qquad < \left(- \pd{B_1(x_1,x_2)}{x_2}\right) \left(- \pd{B_2(x_1,x_2)}{x_1}\right).  
  \end{align}

  Since the blocking probabilities depend $\lambda_i$ and $\mu_i$ only
  through $a_i$, we consider two fictitious providers $P_i^\prime$
  ($i \in \{1,2\}$) with $\mu^{\prime}_1 = \mu^{\prime}_2 = 1$ and
  $\lambda_i^\prime = \lambda_i/\mu_i$ such that
  $B_i^\prime \equiv B_i$. For the providers $P_i^\prime$, we invoke
  Theorem~\ref{thm:monotonicity}, to deduce that $B_{ov}(x_1,x_2)$ is
  a strictly decreasing function of $x_1$ and $x_2$. This means
  \begin{align} \label{eq:b_overall_1} \lambda_1^\prime
    \pd{B_1(x_1,x_2)}{x_1} & < - \lambda_2^\prime \pd{B_2(x_1,x_2)}{x_1}
    \\ \label{eq:b_overall_2} \lambda_2^\prime \pd{B_2(x_1,x_2)}{x_2} & <-
    \lambda_1^\prime \pd{B_1(x_1,x_2)}{x_2}
  \end{align}

  Noting that terms on both sides of \eqref{eq:b_overall_1} and
  \eqref{eq:b_overall_2} are positive, we can multiply the two
  inequalities to obtain \eqref{eq:pareto_1}.

  It is important to note that even though Statement~3 of
  Theorem~\ref{thm:monotonicity} assumes that $\mu_1 = \mu_2,$
  \emph{the present proof does not.}  \endproof

  While Theorem~\ref{thm:pareto} states that the (non-empty) set of
  Pareto-efficient configurations lies on the boundary of the space of
  partial sharing configurations (specifically, in the set
  $\mathcal{X} \setminus \mathcal{X}^o$), it does not provide a
  precise characterization of this set. Interestingly, such a precise
  characterization is possible, which is the goal of the following
  lemma.

\begin{lemma}
  \label{lemma:pareto_structure}
  Under the probabilistic as well as the bounded overflow partial
  sharing models, the set $\hat{\mathcal{P}}$ of Pareto-efficient
  sharing configurations is characterized as follows.
  \begin{enumerate}
  \item If $E(N_1 + N_2,a_1 + a_2) < E(N_i,a_i)$ $\forall\ i,$ then
    there exist uniquely defined constants $\hat{x}_1$ and
    $\hat{x}_2,$ such that for $i = 1,2,$ $\hat{x}_i \in (0,1),$
    \begin{align*}
      B_1(1,\hat{x}_2) &= E(N_1,a_1),\\
      B_2(\hat{x}_1,1) &= E(N_2,a_2).
    \end{align*}
    In this case,
    $$\hat{\mathcal{P}} = \{(x,1):\ x \in(\hat{x}_1,1]\} \cup
    \{(1,x):\ x \in(\hat{x}_2,1]\}.$$
  \item If $E(N_{2},a_{2}) \leq E(N_1 + N_2,a_1 + a_2) < E(N_1,a_1),$
    then there exist uniquely defined constants $\underline{x}_2$ and
    $\bar{x}_2$ satisfying $0 < \underline{x}_2 < \bar{x}_2 \leq 1$
    such that
    \begin{align*}
      B_1(1,\underline{x}_2) &= E(N_1,a_1),\\
      B_2(1,\bar{x}_2) &= E(N_2,a_2).
    \end{align*}
    In this case,
    $$\hat{\mathcal{P}} = \{(1,x):\ x \in
    (\underline{x}_2,\bar{x}_2)\}.$$
  \item If $E(N_{1},a_{1}) \leq E(N_1 + N_2,a_1 + a_2) < E(N_2,a_2),$
    then there exist uniquely defined constants $\underline{x}_1$ and
    $\bar{x}_1$ satisfying $0 < \underline{x}_1 < \bar{x}_1 \leq 1$
    such that
    \begin{align*}
      B_2(\underline{x}_1,1) &= E(N_2,a_2),\\
      B_1(\bar{x}_1,1) &= E(N_1,a_1).
    \end{align*}
    In this case, 
    $$\hat{\mathcal{P}} = \{(x,1):\ x \in (\underline{x}_1,\bar{x}_1)\}.$$
  \end{enumerate}
\end{lemma}

Figure~\ref{fig:pareto} provides a pictorial representation of the set
of Pareto-efficient partial sharing configurations under the three
cases considered in Lemma~\ref{lemma:pareto_structure}. Note that
Case~1 corresponds to settings where full pooling is beneficial to
both providers. Cases~2 and~3 cover the more asymmetric settings,
where exactly one provider (the more congested one) stands to benefit
from full pooling. Lemma~\ref{lemma:pareto_structure} states that in
such cases, the more congested provider pools all of its servers under
any Pareto-efficient sharing configuration. Intuitively, this is
because the asymmetry in the value of servers pooled by each provider
to the other. Indeed, servers pooled by the more congested provider
add less value, since those servers are available for overflow calls
of the less congested provider less often. As a result, mutually
beneficial sharing configurations have the more congested provider
pool more servers than the less congested provider.

The proof of Lemma~\ref{lemma:pareto_structure} is provided in
Appendix~\ref{app:pareto_struc}.

\section{Economics of Partial Sharing}
\label{sec:economics}

The set $\hat{\mathcal{P}}$ of Pareto-efficient configurations
characterized in Section~\ref{sec:pareto} contains all possible
sharing configurations which are minimal for the partial order induced
by the usual relation ``$\leq$'' applied component-wise on the vectors
of possible blocking probabilities. In other words, for every
QoS-stable configuration outside of this set, there exists a
configuration within $\hat{\mathcal{P}}$ that improves the blocking
probability of at least one provider without worsening the blocking
probability for the other. Unfortunately, the configurations within
the Pareto set are not comparable under this component-wise
relation. If we take any two configurations inside this set, then a
configuration that is better for one of the providers will be worse
for the other provider.  Thus, rational providers who want to minimize
their blocking probability will agree that it is beneficial for both
of them to choose a configuration inside the Pareto set rather than
one outside of this set, but will disagree on the choice of the
configuration within the Pareto set.

It is then the natural to ask: Which configuration within the Pareto
set should the two providers choose? Of course, in addition to the
choices within the Pareto set, they could also choose not to
share. This question, in a more general setting, has been investigated
inside the framework of {\it bargaining theory}. In a typical
two-player bargaining problem, two players have to agree upon one
option amongst several. If both agree upon the option, then each
player gets a utility corresponding to this option. On the other hand,
if they fail to arrive at a consensus, then they get a utility
corresponding to that of a {\it disagreement point}. In our setting,
the two players are the two providers who have to choose between the
various configurations. Of course, they could choose not to share with
the other, in which case the blocking probability for each will be
that of the system with no pooling, i.e., the disagreement point is
just the configuration $(0,0)$.

Our aim in this section is to present some of the most common solution
concepts from bargaining theory and apply them to the partial resource
sharing problem under consideration. We also present results of
numerical experiments for different realistic network settings,
highlighting the potential benefits of partial resource pooling in
practice. Note that the discussion in this section applies to both the
partial pooling models defined in Section~\ref{sec:model}.


\subsection{Bargaining solutions}
The usual way to compute a solution of a bargaining problem is to
first fix a set of axioms that a solution must satisfy. Axioms that
appear often (though not necessarily together) are Pareto optimality
(PO), Symmetry (SYM), Scale Invariance (SI), Independence of
Irrelevant Alternatives (IIA) and Monotonicity (MON).

In addition to the axioms, some solution concepts rely on the
convexity of the space of feasible utility pairs in order to guarantee
uniqueness. In the present setting, the utility of a provider is a
strictly decreasing function of its blocking probability. Due to space
constraints, we restrict our attention to the linear case, i.e., the
utility of $P_i$ is taken to be $-B_i,$ where $B_i$ denotes its
blocking probability. Numerical experiments show that this utility
space is not convex. The usual method to overcome this drawback is to
convexify the utility space by considering its convex hull. For our
problem, this could lead to a solution of the form (as an example):
configuration $(k_1, k_2)$ with probability $p$ and $(k'_1, k'_2)$
with probability $(1-p)$. While on an abstract level, a solution in an
extended space is acceptable, in practice its implementation may not
be straightforward. Should the probability $p$ be interpreted as a
fraction of time during which $(k_1,k_2)$ is implemented? If so, at
what time-scale should the changes in configuration occur?

Another method of getting around convexity is to modify the set of
axioms and show that some variation of the solutions concepts for the
convex case satisfy them (see \cite{HP15} and references
therein). These however require some other assumptions on the utility
set such as comprehensiveness\footnote{Comprehensiveness says that for
any vector in the utility set, all vectors that are weakly dominated
by this vector and that weakly dominate the disagreement point are
also in the utility set.} which is again difficult to verify in our
setting.

We now apply four bargaining solutions from the literature to our
partial pooling model. These are the Nash, Kalai-Smorodinsky,
egalitarian and utilitarian bargaining solutions. The main result in
this section shows the uniqueness of the Kalai-Smorodinsky and the
egalitarian solutions without calling upon the standard arguments of
convexity or comprehensiveness. The proof is based upon monotonicity
properties highlighted in Section~\ref{sec:model}.

For the bargaining solutions in this section, we assume that the
utility of each provider is the negative of its blocking
probability. In some situations, it may be more meaningful to take the
negative logarithm of the blocking probability as the utility of a
provider. We give the logarithmic variants of the Nash,
Kalai-Smorodinsky, and the egalitarian solutions in Appendix
\ref{app:logbs}.

\subsubsection*{Nash bargaining solution}

The first concept we present was proposed by Nash in the seminal paper
\cite{N50}.
\begin{definition}
  A partial sharing configuration $(x_1^*,x_2^*)$ is \emph{ Nash
    bargaining solution (NBS)}, if the partial sharing configuration
  satisfies the following condition: $(x_1^*,x_2^*) = $
$$\argmax_{x \in [0,1]^2}  \left[B_1(0,0) - B_1(x_1,x_2) \right]_+ \left[ B_2(0,0) - B_2(x_1,x_2) \right]_+  . $$
\end{definition}
Here $[z]_+ := max(z,0)$ denotes the positive part of $x.$ At the NBS,
the players are maximizing the product of the individual utilities
relative to the disagreement point\footnote{Here, {\it relative} means
  upon subtracting the utilities at the disagreement point.}. Clearly,
any maximizer would lie in the set of $\hat{\mathcal{P}}.$ However,
the drawback of the NBS for our problem is that the utility space is
not convex (observed in numerical experiments) which implies that the
NBS may not be unique.

\subsubsection*{Kalai-Smorodinsky bargaining solution}

One of criticisms of the NBS is the axiom of IIA which may not hold in
practice. In \cite{KS75}, Kalai and Smorodinsky replaced IIA with MON
and obtained the following solution concept.
\begin{definition}
  A partial sharing configuration $(x_1^*,x_2^*)$ is a \emph{
    Kalai-Smorodinsky bargaining solution (KSBS)}, if
  $(x_1^*,x_2^*) \in \hat{\mathcal{P}}$ and satisfies
\begin{align*}
  \frac{B_1(0,0) - B_1(x_1,x_2)}{ B_2(0,0) - B_2(x_1,x_2)} &~= \frac{B_1(0,0)
  - \min\limits_{y \in [0,1]^2} B_1(y_1,y_2)}{ B_2(0,0) - \min\limits_{y \in [0,1]^2} B_2(y_1,y_2)}. 
 \end{align*}
\end{definition}
At a KSBS solution the ratio of relative utilities of the providers is
equal to the ratio of their maximal relative utilities.
For our problem, the following results guarantees
uniqueness of the solution which could be make it potentially more
attractive than the NBS.

\begin{theorem}
\label{lemma:KSBS}
For the bounded overflow sharing model, the KSBS is unique.
\end{theorem}

\proof[Proof of Theorem~\ref{lemma:KSBS}] Define the following
functions.
\begin{align*}
f(x_1,x_2) := &~\frac{B_1(0,0) - B_1(x_1,x_2)}{ B_2(0,0) - B_2(x_1,x_2)} \\
KS := &~\frac{B_1(0,0) - \min\limits_{y \in [0,1]^2} B_1(y_1,y_2)}{ B_2(0,0) - \min\limits_{y \in [0,1]^2} B_2(y_1,y_2)}
\end{align*}
From the Statements~1 and~2 of Theorem \ref{thm:monotonicity}, we get
\begin{align*}
  \min\limits_{y \in [0,1]^2} B_1(y_1,y_2) &=  B_1(0,1),\\
  \min\limits_{y \in [0,1]^2} B_2(y_1,y_2) &=  B_2(1,0). 
\end{align*}
i.e., each provider gets the maximum benefit when it pools none of its
servers and the other provider pools all of its servers.

It is easy to see that $0< KS < \infty. $
Consider the three cases for $\hat{\mathcal{P}}$ from
Lemma~\ref{lemma:pareto_structure}.
 
\noindent{\bf{Case 1:}}  $E(N_1 + N_2,a_1 + a_2) < E(N_i,a_i)$ $\forall\ i$  \\
Sweeping the (topologically one-dimensional) Pareto-frontier clockwise
from $(\hat{x}_1,N_2)$ to $(N_1,\hat{x}_2),$ it is easy to see that
$f$ is strictly decreasing and continuous, with
\begin{align*}
  &\lim_{x_1 \da \hat{x}_1} f(x_1,1) = \infty, \\
  &\lim_{x_2 \da \hat{x}_2} f(1,x_2) = 0. 
\end{align*}
There is thus a unique point on the Pareto-frontier that satisfies the
KSBS condition.

\noindent{\bf{Case 2:}} $E(N_{2},a_{2}) \leq E(N_1 + N_2,a_1 + a_2) < E(N_1,a_1)$\\
As before, sweeping the Pareto-frontier clockwise from
$(1,\bar{x}_2)$ to $(1,\underline{x}_2),$ it is easy to see that
$f$ is strictly decreasing and continuous, with
\begin{align*}
  &\lim_{x_2 \ua \bar{x}_2} f(1,x_2) = \infty, \\
  &\lim_{x_2 \da \underline{x}_2} f(1,x_2) = 0.
\end{align*}
There is thus a unique point on the Pareto-frontier that satisfies the
KSBS condition.

\noindent{\bf{Case 3:}} $E(N_{1},a_{1}) \leq E(N_1 + N_2,a_1 + a_2) < E(N_2,a_2)$ \\
The argument here is analogous to that for the above cases.  \endproof
 
\subsubsection*{Egalitarian solution}
The next solution concept we present was also proposed by Kalai \cite{K77}. It satisfies PO, SYM, IIA, and MON but violates SI. It captures the sharing configuration in which the gains relative to the disagreement solution for both the providers is the same. 
\begin{definition}
  A partial sharing configuration $(x_1^*,x_2^*)$ is an \emph{
    egalitarian solution (ES)}, if
  $(x_1^*,x_2^*) \in \hat{\mathcal{P}}$ and satisfies
  \begin{align*}
    B_1(0,0) - B_1(x_1,x_2) =  B_2(0,0) - B_2(x_1,x_2).
  \end{align*}
\end{definition}
Under an ES, the providers will see the same amount of improvement in
their blocking probabilities relative to the no-sharing option. The
following result shows that the ES is unique. Its proof follows
similar lines as the proof of Theorem~\ref{lemma:KSBS}.

\begin{lemma}
\label{lemma:ES}
For the bounded overflow sharing model, the ES is unique.
\end{lemma}
\proof[Proof of Lemma \ref{lemma:ES}] The argument in the proof of
Theorem~\ref{lemma:KSBS} applies as is here, except that the constant
$KS$ is replaced by~1.
\endproof
An interesting property of the ES is that if the standalone blocking
probabilities of the two providers are identical, that the ES
corresponds to complete pooling.

\begin{lemma}
  \label{lemma:es}
  If $E(N_1,a_1) = E(N_2,a_2),$ then the ES lies at $(1,1).$
\end{lemma}
\proof[Proof of Lemma \ref{lemma:es}]
  We invoke the following well known property of the Erlang-B formula.
  $$E(N_1 + N_2, a_1 + a_2) < \frac{a_1}{a_1 + a_2} E(N_1,a_1) + \frac{a_2}{a_1 + a_2} E(N_2,a_2).$$
  If $E(N_1,a_1) = E(N_2,a_2),$ it follows then that
  $$E(N_1 + N_2, a_1 + a_2) < E(N_1,a_1) = E(N_2,a_2),$$
  implying that the set $\hat{\mathcal{P}}$ of Pareto-efficient
  configurations includes $(1,1)$ (see
  Lemma~\ref{lemma:pareto_structure}).
  
  Further, is $E(N_1,a_1) = E(N_2,a_2),$ then the ES clearly satisfies
  $B_1(x_1,x_2) = B_2(x_1,x_2).$ However, from the monotonicity
  properties of the blocking probabilities, $(1,1)$ is the only point
  in $\hat{\mathcal{P}}$ that satisfies this property.
\endproof

\subsubsection*{Utilitarian solution}
The last solution concept is that of utilitarian bargaining solution
(see, e.g., \cite{T81}). It minimizes the blocking probability of the
customers as a whole without distinguishing them according the
provider to which they subscribe. It captures the greatest good to the
system.  The advantage is that it is a concept that is easy for
customers to identify with. On the other hand, the axioms of SI and
MON are violated. Nonetheless, the violation of SI does not seem to be
problematic when the utilities are blocking probabilities. Indeed,
there is a unique natural scale on which the blocking probability
satisfies the axioms that define a probability measure.

\begin{definition}
  A partial sharing configuration $(x_1^*,x_2^*)$ is a \emph{
    utilitarian bargaining solution} (US) if it satisfies
$$ \argmin_{k \in \boldsymbol{C}(\hat{\mathcal{P}})}  \frac{\lambda_1}{\lambda_1 + \lambda_2} B_1(x_1x_2) + \frac{\lambda_2}{\lambda_1 + \lambda_2}B_2(x_1,x_2)  . $$
\end{definition}

Here, $\boldsymbol{C}(\hat{\mathcal{P}})$ denotes the closure of
$\hat{\mathcal{P}}.$ We relax the above minimization to be over
$\boldsymbol{C}(\hat{\mathcal{P}})$ instead of over the open set
$\hat{\mathcal{P}}$ because in some cases, it turns out that the
solution lies on the boundary. Assuming that the average call holding
time for both providers is identical, the utilitarian solution is
unique and can be characterized precisely.

\begin{lemma}
  \label{lemma:us}
  If $\mu_1 = \mu_2,$ under the bounded overflow model, the US is
  characterized as follows.\footnote{We use the notation from
    Lemma~\ref{lemma:pareto_structure}.}
  \begin{enumerate}
  \item If $E(N_1 + N_2,a_1 + a_2) < E(N_i,a_i)$ $\forall\ i,$ then
    the US is $(1,1)$
  \item If $E(N_{2},a_{2}) \leq E(N_1 + N_2,a_1 + a_2) < E(N_1,a_1),$
    then the US is $(1,\bar{x}_2)$
  \item If $E(N_{1},a_{1}) \leq E(N_1 + N_2,a_1 + a_2) < E(N_2,a_2),$
    then the US is $(\bar{x}_1,1)$
  \end{enumerate}
\end{lemma}
We omit the proof of Lemma~\ref{lemma:us}, since it is direct
consequence of Statement~3 of Theorem~\ref{thm:monotonicity}. Another
quick observation is that when the standalone blocking probabilities
are matched, the utilitarian solution, like the egalitarian solution,
corresponds to full pooling.
\begin{corollary}
  \label{corr:us}
  If $E(N_1,a_1) = E(N_2,a_2),$ then the US lies at $(1,1).$
\end{corollary}
\proof[Proof of Corollary \ref{corr:us}.]
  As was argued in the proof of Lemma~\ref{lemma:es}, if
  $E(N_1,a_1) = E(N_2,a_2),$ then 
  $$E(N_1 + N_2, a_1 + a_2) < E(N_1,a_1) = E(N_2,a_2).$$
  The statement of the corollary now follows from
  Lemma~\ref{lemma:us}.
\endproof

While the utilitarian solution is the most efficient, in that is
minimizes the overall blocking probability, it may not be
fair. Indeed, under Cases~2 and~3 of Lemma~\ref{lemma:us} above, one
of the providers (the less congested provider) sees no reduction in
its blocking probability relative to the disagreement point.


\subsection{Numerical examples}

In this section, we present numerical results illustrating the various
bargaining solutions under realistic system settings. The goal of this
section is two-fold: to demonstrate the benefits of partial resource
pooling to the two providers, and to illustrate differences between
the different bargaining solutions. Due to space constraints, we are
only able to consider two network settings. Also, restrict our
attention in this section to the bounded overflow sharing model; we
represent the bargaining solution as $(k_1^*, k_2^*),$ where $k_i^* =
N_i x_i^*.$

\begin{table}[h]
\centering
\caption{Different bargaining solutions for the case
  $N_1 = N_2 = 100,$ the standalone blocking probabilities of $P_1$
  and $P_2$ being 6\% and 1\%, respectively. Mean call holding times
  are assumed to be equal for both providers.}
\label{tb:bo_same_size}
\begin{tabular}{|l|c|c|c|c|}
  \hline
  Bargaining & $k_1^*$ & $k_2^*$ & $B_1$ & $B_2$\\
  solution  &&&&\\
  \hline
  US & 100 & 13.1 & 1. 73\%& 1\% \\
  KSBS & 100 & 6 & 3.39\% & 0.63 \% \\
  NBS & 100 & 5.5 & 3.6\% & 0.6 \% \\
  ES & 100 & 1.35 & 5.36\% & 0.36\% \\
  \hline
\end{tabular}
\end{table}

First we consider a scenario where the two providers have the same
number of servers, but differ with respect to their standalone
blocking probabilities. Specifically, we set $N_1 = N_2 = 100,$ with
$E(N_1,a_1) = 0.06$ (6\%), $E(N_1,a_1) = 0.01$ (1\%), and
$\mu_1 = \mu_2 = 1.$ Clearly $P_1$ is the more congested provider. The
different bargaining solutions for this scenario are summarized in
Table~\ref{tb:bo_same_size}. As expected, the more congested provider
$P_1$ pools all its servers under all bargaining solutions. Moreover,
the `efficient' utilitarian solution is the most beneficial for $P_1,$
while not providing any benefit to $P_2.$ At the other extreme, ES is
the most pessimal, since it enforces the same reduction in blocking
probability, even though the scope for reduction is much less for
$P_2.$ KSBS and NBS result in intermediate contributions by $P_2,$ and
result in a substantial benefits for both $P_1$ and $P_2;$ indeed,
these configurations result in a roughly 40\% reduction in the
blocking probability of each provider.

\begin{table}[h]
\centering
\caption{Different bargaining solutions for the case $N_1 = 200,$
  $N_2 = 50,$ both providers having a standalone blocking probability
  of 5\%. Mean call holding times are assumed to be equal for both
  providers.}
\label{tb:bo_same_blockprob}
\begin{tabular}{|l|c|c|c|c|}
  \hline
  Bargaining & $k_1^*$ & $k_2^*$ & $B_1$ & $B_2$\\
  solution  &&&&\\
  \hline
  US & 200 & 50 & 3.33\% & 3.33\% \\
  ES & 200 & 50 & 3.33\% & 3.33\% \\
  NBS & 200 & 9.5 & 3.36\% & 3.19\% \\
  KSBS & 200 & 8 & 3.56\% & 2.99\% \\
  \hline
\end{tabular}
\end{table}

Next, we consider a scenario where the two providers differ in size,
but are matched with respect to standalone blocking
probability. Specifically, we set $N_1 = 200,$ $N_2 = 50,$
$E(N_1,a_1) = E(N_2,a_2) = 0.05,$ and $\mu_1 = \mu_2 = 1.$ The results
are summarized in Table~\ref{tb:bo_same_blockprob}. As expected, the
US as well as the ES correspond to complete pooling (see
Lemma~\ref{lemma:es} and Corollary~\ref{corr:us}); this results in
both providers seeing a blocking probability of 3.33\%. On the other
hand, the NBS as well as the KSBS, the smaller provider ($P_2$) pools
fewer servers. As a result, the smaller provider achieves an even
lower blocking probability under KSBS/NBS, at the expense of a higher
blocking probability for the larger provider (compared to the full
pooling under US/ES). As before, it is important to note that partial
resource pooling offers the possibility of substantially lower
blocking probability for both providers.

\section{Large System Limits: Square root scaling}
\label{sec:large}
The computational complexity of the exact steady-state blocking
probability increases as the number of circuits becomes large
\cite{Kelly91}. As a result, approximations can turn out to be helpful
for their tractability as well as their ability to provide insights
into the complex dependencies between the blocking probabilities and
the system parameters. The goal of this section is to obtain large
system approximations for the blocking probabilities under the bounded
overflow partial pooling model.\footnote{A parallel development for
  the probabilistic sharing model is possible, which we omit due to
  space constraints.}


Large system approximations have always been an
integral part the literature on queueing theory. Depending
upon the parameters of systems, these limits can take different forms
such as mean-field \cite{Mukho15}, Quality and Efficiency Driven
\cite{HW81}, or Non-degenerate Slowdown \cite{Atar12} limits.

\subsection{QED scaling regime}

For our resource sharing model with blocking, the most relevant limit
is the quality-efficiency-driven (QED) regime (a.k.a. ``square-root
staffing'' regime, Halfin-Whitt regime). While it is now commonly
known under these names, it had already been investigated by Erlang
himself\footnote{See the paper "On the rational determination of the
  number of circuits" in \cite{BHA48}.} and Jagerman as well
\cite{Jagerman74}. The traditional QED regime applies to system with a
single provider, and is defined as follows. Let $N$ be the number of
circuits with the provider and $a$ be the offered load. We say that
$f(t) \sim g(t)$ as $t \ra \infty$ if
$\lim_{t \ra \infty} \frac{f(t)}{g(t)} = 1.$
\begin{lemma}[\cite{Jagerman74}]
\label{lemma:ErlangB_sqroot}
Let $a = N + \beta \sqrt{N} + o(\sqrt{N})$. Then, 
\[
  E(N,a) \sim \frac{1}{\sqrt{N}}\dfrac{\phi(\beta)}{(1-\Phi(\beta))}
  \quad \text{ as } N \ra \infty.
\]
\end{lemma}
Here, $\phi(\cdot)$ and $\Phi(\cdot)$ denote, respectively, the
probability density function and the cumulative distribution function,
corresponding to the standard Gaussian distribution. Note that under
the QED regime, the margin between the offered load and the number of
servers is of the order of the square root of the number of
servers. In many settings, the QED regime is known to be the right
balance between quality (i.e., QoS) and efficiency (i.e., server
provisioning costs); see, for example, \cite{HW81,Borst04}. For the
M/M/N/N loss system, Lemma~\ref{lemma:ErlangB_sqroot} states that the
steady state blocking probability decays as $\Theta(1/\sqrt{N})$ as
$N \ra \infty.$

We define the QED scaling regime for our model with two providers as
follows. For fixed $\alpha_i > 0$ and $\beta_i \in \R$, let
 \begin{align}
   N_i & =  \al{i} N, \label{eqn:N_lsl}\\
   a_i & =  N_i + \beta_i\sqrt{N_i} + o(\sqrt{N_i}). \label{eqn:a_lsl}
 \end{align}
 Here, $N$ is the scaling parameter that is common to both
 providers. \eqref{eqn:N_lsl} states that the number of servers of
 each provider grow proportionately with the scaling
 parameter. \eqref{eqn:a_lsl} states that the offered load
 corresponding to each provider scales as per the QED (square-root
 staffing) rule.


 Before deriving the blocking probabilities for the different partial
 sharing configurations, we first look at two special cases for which
 these probabilities can be derived directly from
 Lemma~\ref{lemma:ErlangB_sqroot}.
 With no resource pooling, both the providers are decoupled, and for
 large $N$, the steady state blocking probability of Provider~$ i $
 can be computed using Lemma \ref{lemma:ErlangB_sqroot} to be
\begin{align*}
  B_i \sim \frac{1}{\sqrt{N_i}} \dfrac{\phi(\beta_i)}{(1-\Phi(\beta_i))}.
\end{align*}
The second special case is that of full resource pooling. Here, the
system acts as a single provider with $( N_1 + N_2) $ servers/circuits
and offered load of $(a_1+a_2) $. By simple calculations we can see
the system under full pooling also satisfies the square root scaling
set up. So, the steady state blocking probability for both the
providers is given as
\begin{align*}
  B_{full} \sim \frac{1}{\sqrt{N(\alpha_1 +\alpha_2)}}\frac{\phi\left( \dfrac{\beta_1 \sqrt{\alpha_1}+\beta_2 \sqrt{\alpha_2}}{\sqrt{\alpha_1+\alpha_2}}\right)}{1- \Phi\left(\dfrac{\beta_1 \sqrt{\alpha_1}+\beta_2 \sqrt{\alpha_2}}{\sqrt{\alpha_1+\alpha_2}}\right)}.
\end{align*}

Now, we present the square-root scaling set up for partial sharing
configurations. For $\g{i} \geq 0,$ we scale the sharing parameters as
\begin{align}
  k_{i} = \g{i} \sqrt{N_i} + o(\sqrt{N_i}).
  \label{eqn:sqrt_staff}
\end{align}
Note that the number of pooled servers for $P_i$ is scaled in
proportion to $\sqrt{N_i}.$ It turns out that for the system scaling
defined by \eqref{eqn:N_lsl}--\eqref{eqn:a_lsl}, this is the only
meaningful manner of scaling the partial sharing parameters. Indeed,
if $k_i = o(\sqrt{N_i}),$ then the large system limits correspond to
$P_i$ pooling no servers, and if $k_i = \omega(\sqrt{N_i}),$ then the
large system limits correspond to $P_i$ pooling all its
servers. Intuitively, this is because on the diffusion scale defined
by \eqref{eqn:N_lsl}--\eqref{eqn:a_lsl}, the number of overflow calls
as well as the number of free servers of each provider evolve (in
time) on the $\sqrt{N}$ scale.

To summarize, the QED regime we consider is defined by
\eqref{eqn:N_lsl}--\eqref{eqn:sqrt_staff}. Our main result in this
section gives the relationship between the asymptotic blocking
probability for each provider and the various parameters of the
system, namely, the sharing parameters $(\gamma_1,\gamma_2),$ the
square-root staffing margins $(\beta_1,\beta_2)$, and the relative
sizes of the two providers $(\alpha_1,\alpha_2).$

\subsection{Blocking probability asymptotics}

Having defined our QED scaling regime, we now derive large system
asymptotics of the blocking probabilities. Our results are summarized
in the following theorem.

\begin{figure}[h]
    \centering
    \includegraphics[width=0.4\textwidth]{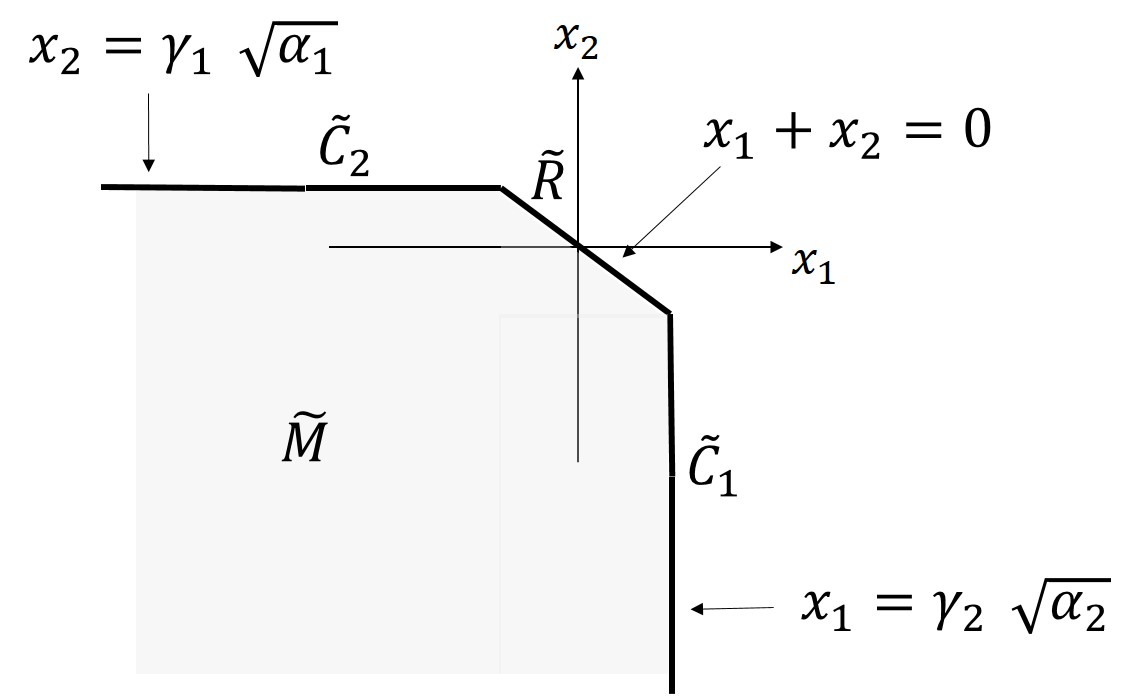}
    \caption{Geometric interpretation of blocking probability
      asymptotics under QED}
    \label{fig:qed}
\end{figure}

\begin{theorem}
  \label{thm:qed}
  Under the bounded overflow sharing model, for the scaling regime
  defined by \eqref{eqn:N_lsl}--\eqref{eqn:sqrt_staff}, the steady
  state blocking probability of Provider~$i$ for large $N$ is given as
   \begin{equation*}
     B_i(\gamma_1,\gamma_2) \sim \frac{1}{\sqrt{N}} \frac{\tilde{A}_i}{\tilde{G}},
   \end{equation*}
   where
   \begin{align*}
   		\tilde{A}_i = & \frac{\phi\left(\gamma_{-i}\sqrt{\frac{\alpha_{-i}}{\alpha_i}} - \beta_i\right)}{\sqrt{\alpha_i}} \Phi \left(-\gamma_{-i} - \beta_{-i} \right)\\
   		& + \frac{1}{\sqrt{\alpha_1 \alpha_2}} \int\limits_{-\gamma_1 \sqrt{\alpha_1}}^{\gamma_2 \sqrt{\alpha_2}} \phi\left(\frac{x}{\sqrt{\alpha_1}} - \beta_1 \right) \phi\left(\frac{-x}{\sqrt{\alpha_2}} - \beta_2\right) \,dx, \\
   		\tilde{G} = & \Phi\left(\gamma_1\sqrt{\frac{\alpha_1}{\alpha_2}} - \beta_2\right) \Phi\left(-\gamma_1  - \beta_1 \right) \\
   		& + \frac{1}{\sqrt{\alpha_1}} \int\limits_{-\gamma_1 \sqrt{\alpha_2}}^{\gamma_2 \sqrt{\alpha_1}} \phi\left(\frac{x}{\sqrt{\alpha_1}} - \beta_1 \right) \Phi\left(\frac{-x}{\sqrt{\alpha_2}} - \beta_2\right) \,dx.
   \end{align*}
 \end{theorem}
 
 Even though the expressions for $\tilde{A}_i$ and $G$ in the
 statement of Theorem~\ref{thm:qed} look complicated, they have a
 simple geometric interpretation. To see this, define the following
 sets in $\mathbb{R}^2.$
 \begin{align*}
   \tilde{M}&:= \left\{(x_1,x_2):\ \substack{x_1 \leq \gamma_2
              \sqrt{\alpha_2} \\x_2 \leq \gamma_1 \sqrt{\alpha_1} \\x_1 + x_2
   \leq 0} \right\}, \\
   \tilde{R}&:= \left\{(x_1,x_2):\ \substack{x_1 \leq \gamma_2
              \sqrt{\alpha_2} \\x_2 \leq \gamma_1 \sqrt{\alpha_1} \\x_1 + x_2
   = 0} \right\}, \\
   \tilde{C}_1 &:= \left\{(x_1,\gamma_1 \sqrt{\alpha_1}):\ x_1 \leq -\gamma_1 \sqrt{\alpha_1})\right\}, \\
   \tilde{C}_2 &:= \left\{(\gamma_2 \sqrt{\alpha_2},x_2):\ x_2 \leq -\gamma_2 \sqrt{\alpha_2})\right\}. \\
 \end{align*}
 These sets are depicted in Figure~\ref{fig:qed}. Note that
 $\tilde{M}$ is the shaded pentagonal region, and $\tilde{R},$
 $\tilde{C}_1,$ and $\tilde{C}_2,$ represent the diagonal, right, and
 upper boundaries of $\tilde{M},$ respectively. Now, define
 independent Gaussian random variables $Z_1$ and $Z_2,$ such that
 $Z_i$ has mean $\beta_i \sqrt{\alpha_i},$ and variance $\alpha_i.$
 Let $f_{Z_i}(\cdot)$ denote the probability density function
 corresponding to $Z_i.$ With this notation, it is not hard to show
 that
 \begin{align}
   \tilde{A}_i &= \int_{\tilde{R}} f_{Z_1}(x_1) f_{Z_2}(x_2) + \int_{\tilde{C}_i} f_{Z_1}(x_1) f_{Z_2}(x_2), \nonumber \\
   \label{eq:geometricG}
  \tilde{G} &= \iint_{\tilde{M}} f_{Z_1}(x_1) f_{Z_2}(x_2) 
 \end{align}
 This means that $\tilde{A}_i$ is the line integral of the joint
 density function of $Z_1$ and $Z_2$ over $\tilde{R} \cup
 \tilde{C}_i,$ and $\tilde{G}$ is the integral of the same joint
 density function over the region $\tilde{M}$ (in other words,
 $\tilde{G}$ is the probability that the random vector $(Z_1,Z_2)$
 takes a value in $\tilde{M}$).

Theorem~\ref{thm:qed} yields a computationally tractable approximation
for the blocking probabilities under the BO partial sharing model,
which is asymptotically accurate under the QED regime. In particular,
note that the computational complexity of the approximation is
invariant to the system size, making it particularly attractive when
the number of servers is large. In the remainder of this section, we
evaluate the accuracy of the large system approximation under
realistic network settings. The proof of Theorem~\ref{thm:qed} is
presented in Appendix~\ref{app_thm:qed}.


\subsection{Accuracy of large system approximation}


\ignore{First, we describe how the we apply the approximation to a
  given (finite) system setting, characterized by the
  tuple $$(N_1,N_2,a_1,a_2,k_1,k_2).$$ Note that the large system
  approximation is in turn characterized by the
  tuple $$(\alpha_1,\alpha_2,\beta_1,\beta_2,\gamma_1,\gamma_2).$$ We
  set $\alpha_1 = 1,$ and $\alpha_2 = \frac{N_2}{N_1}.$
}

We consider the case $N_1 = N_2 = N,$ with the standalone blocking
probabilities of Provider~1 and~2 being 0.05 and 0.01,
respectively. We vary $N$ and compute the error between the exact
blocking probability and the large system approximation. In
Figures~\ref{fig:acc_b1} and~\ref{fig:acc_b2}, we plot the minimum and
maximum of the ratio between the exact and the approximate blocking
probability over the set of all feasible partial sharing vectors
$(k_1,k_2).$ We note that the approximation becomes increasingly
accurate as the system size grows, the error being under 8\% for $N =
200.$ It is noteworthy that we have not scaled the system under the
QED regime in this example -- we are simply fixing the standalone
blocking probabilities to realistic values, and growing the number of
servers to moderate levels. Despite this, our approximation, which was
developed using the QED scaling regime, is quite accurate.

For $N = 200,$ we also plot the Pareto frontier, computed using the
exact blocking probability expression and the large system
approximation; see Figure~\ref{fig:large_system_pareto}. We note that
the two sets are quite close, suggesting that one could potentially
use the large system approximation to determine meaningful bargaining
solutions.

\begin{figure*}[t!]
    \centering
    \begin{subfigure}[t]{0.3\textwidth}
        \centering
        \includegraphics[scale=0.4]{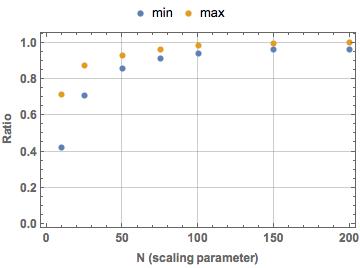}
        \caption{Ratio of exact and approximate for blocking
          probability of Provider~1}
 		 \label{fig:acc_b1}
    \end{subfigure}%
    ~ 
    \begin{subfigure}[t]{0.3\textwidth}
        \centering
        \includegraphics[scale=0.4]{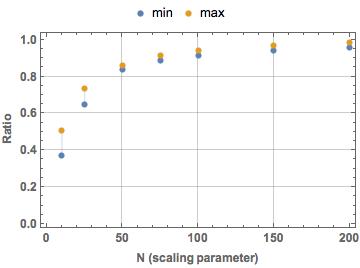}
        \caption{Ratio of exact and approximate for blocking
          probability of Provider~2}
  		\label{fig:acc_b2}
    \end{subfigure}
    ~
    \begin{subfigure}[t]{0.3\textwidth}
        \centering
        \includegraphics[scale=0.4]{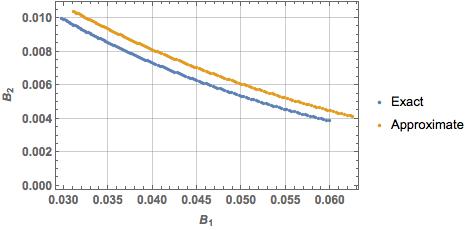}
        \caption{Pareto frontier computed using exact and approximate
          blocking probabilities}
  		\label{fig:large_system_pareto}
    \end{subfigure}
    \caption{Numerical accuracy of large system approximation}
    \label{fig:large_system}
\end{figure*}

\section{Discussion}
\label{sec:discuss}

\begin{figure}[ht]
    \centering
    \includegraphics[width=0.2\textwidth]{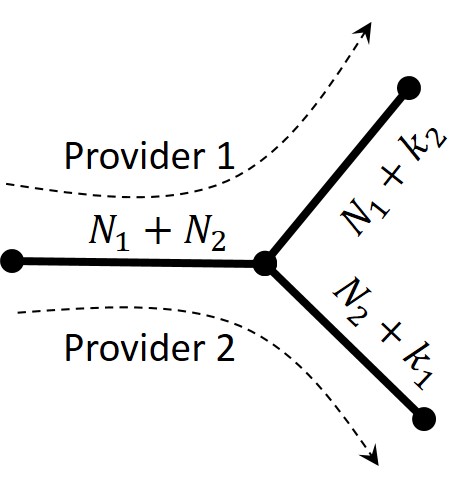}
    \caption{A circuit multiplexed network with three links (with
      capacities as marked) and two routes (correponding to the two
      providers)}
    \label{fig:cktnet}
\end{figure}

We conclude with a discussion on some analogies from circuit
multiplexed networks and possible generalizations.

\subsubsection*{Circuit Multiplexed Network Analogs}
When the $k_i$ are integers, the bounded overflow model, the state
space and the stationary distribution will be the same as the circuit
multiplexed network of with three links and two routes as shown in
Fig.~\ref{fig:cktnet}. With this representation the reduced load
approximation method of, e.g., \cite{Kelly88} may also be used to
calculate the blocking probabilities. However, it is not
computationally simpler than the exact formula of
Lemma~\ref{lemma:bound_sharing}. 

As was mentioned in Section~\ref{sec:model}, there is a superficial
similarity between the BO model and trunk reservation. Trunk
reservation has been used in circuit multiplexed networks to give
preference to direct route calls over alternate route calls. This is
done by reserving the `last $k$ circuits' for direct route calls. This
means that on a link with $N$ circuits, alternate route calls are not
admitted when the number of idle circuits is less than or equal to
$k.$ Exact models for trunk reservation are hard to analyze and
asymptotic analyses, e.g., \cite{Reiman89}, are among analytical
techniques that have been used to model trunk reservation.

\ignore{
\subsubsection*{On Other Models for Partial Pooling}
Several extensions and generalisations to the partial pooling model
are possible. A possible generalisation is to associate a probability
$x_{-i}(n)$ to admit an overflow call when $n_i=N_i+n.$ Under this
generalization, we see that the BO model corresponds to 
\begin{displaymath}
  x_{-i}(n) = \begin{cases}
    1 & \mbox{for $n_i \leq N_i+ \lfloor k_{-i} \rfloor, $} \\
    \{k_{-i}\} & \mbox{for $n_i=N_i+k_{-i},$} \\
    0 & \mbox{for $n_i > N_i+ \lceil k_{-i} \rceil.$}
  \end{cases}
\end{displaymath}
Under this generalization, we could consider a simple probabilistic
sharing model where $x_{i}(n)=p_i$ for all $n.$ For this model too our
results of Sections~\ref{sec:pareto} and \ref{sec:economics} would
hold if the monotonicity of Theorem~\ref{thm:monotonicity} is holds.
While intuition and our extensive empirical results indicate that a
Theorem~\ref{thm:monotonicity} probably also applies to probabilistic
sharing, a proof has been elusive.
}

\subsubsection*{Future Work} 
Several extensions of partial pooling to models in extant literature
are possible. Erlang-C or waiting models is an obvious immediate
model. Partial inventory pooling is another application that could be
explored. A third application would be in single server systems where
the quality of service for each customer is a decreasing function of
the number of active calls like in discriminatory processor sharing
systems and CDMA systems. Here the providers would share a part of the
servers' capacity and the sharing configuration could depend on the
service degradation as a function of the number of active calls and
the server capacity. These are currently being explored.

\bibliographystyle{unsrt}
\bibliography{spectrum-sharing} 


\appendix

\section{Proof of Theorem~\ref{thm:monotonicity} for the bounded
  overflow sharing model}
\label{app:monotonicity_bo}

This section is devoted to the proof of Theorem~\ref{thm:monotonicity}
for the bounded overflow sharing model. The following well known
properties of the Erlang-B formula will be useful.
\begin{lemma}
  \label{lemma:ErlangB-LB}
  $E(N,a)$ is a strictly decreasing function of $N.$ Moreover, $E(N,a)
  > 1 - \frac{N}{a}.$
\end{lemma}
We also state the following lemma which will be invoked repeatedly in
the proof.
\begin{lemma}
  \label{lemma:det_B_bounds}
  Under the bounded overflow sharing model, if $k_i \in
  \{1,2,\cdots,N_i\} \ \forall i,$
  \begin{align*}
    E(N_i+{k_{-i}},a_i) < B_i^{(bo)}(k_1,k_2) <
    E(N_i-{k_{i}},a_i) \quad (i \in \{1,2\}).
  \end{align*}
\end{lemma}
The proof is elementary and is omitted.

\subsection{Proof of Statements~1 and~2}

The steady state blocking probability of Provider~1 can be expressed
as follows.
\begin{align*}
B_1^{(bo)}(k_1,k_2) = \frac{m_1+ u_1\fr{k_1}+v_1\fr{k_2}}{d+u\fr{k_1}+v\fr{k_2}}
\end{align*}
Here, 
\begin{align*}
  m_1 &= \sum_{i=N_1-\floor{k_1}}^{N_1 + \floor{k_2}} \frac{a_1^{i}}{i!}\frac{a_2^{(N_1+N_2 - i)}}{(N_1+N_2 - i)!} + \frac{a_1^{(N_1+\floor{k_2})}}{(N_1+\floor{k_2})!} \sum_{j=0}^{N_2 - \floor{k_2}-1} \frac{a_2^j}{j!},\\
  u_1 &= \frac{a_1^{(N_1-\ceil{k_1})}}{(N_1-\ceil{k_1})!}\frac{a_2^{(N_2+\ceil{k_1})}}{(N_2+\ceil{k_1})!},\\
  v_1 &= \left(1 - \frac{N_1 + \ceil{k_2}}{a_1}\right) \frac{a_1^{(N_1+\ceil{k_2})}}{(N_1+\ceil{k_2})!}  \sum_{j=0}^{N_2-\ceil{k_2}} \frac{a_2^j}{j!} ,
\end{align*}
\begin{align*}
  d &= \sum_{(i,j):\ \substack{i \leq N_1 + \floor{k_2} \\ j \leq N_2 + \floor{k_1} \\ i+j \leq N_1 + N_2}} \frac{a_1^{i}}{i!} \frac{a_2^j}{j!}, \\
  u &= \frac{a_2^{(N_2+\ceil{k_1})}}{(N_2+\ceil{k_1})!} \sum_{i = 0}^{N_1 - \ceil{k_1}} \frac{a_1^{i}}{i!}, 
  v = \frac{a_1^{(N_1+\ceil{k_2})}}{(N_1+\ceil{k_2})!} \sum_{j=0}^{N_2-\ceil{k_2}} \frac{a_2^j}{j!}.
\end{align*}
Since $B_1^{(bo)}(k_1,k_2)$ is continuous in its arguments, it
suffices to show that for non-integer
$(k_1,k_2),$ $$\pd{B_1^{(bo)}(k_1,k_2)}{k_2} < 0,\quad
\pd{B_1^{(bo)}(k_1,k_2)}{k_1} > 0.$$ Accordingly, in the remainder of
the proof, we make the assumption that $\fr{k_1},\fr{k_2} \neq 0.$

We now prove that $\pd{B_1^{(bo)}(k_1,k_2)}{k_2} < 0.$ An elementary
calculation yields
\begin{align*}
\pd{B_1^{(bo)}(k_1,k_2)}{k_2} = \frac{(v_1 d -
  vm_1)+(v_1u-vu_1)\fr{k_1}}{(G')^{2}}
\end{align*}
It now suffices to show that each term in the numerator above is
negative. To see that the first term is negative, note that 
\begin{align*}
  v_1d-vm_1 &= vd\biggl( \frac{v_1}{v} - \frac{m_1}{d}\biggr)\\
  &= vd\left(\left(1 - \frac{N_1 + \ceil{k_2}}{a_1}\right) - B_1^{(bo)}(\floor{k_1},\floor{k_2}) \right)\\
  &< vd \left(E(N_1 + \ceil{k_2},a_1) - B_1^{(bo)}(\floor{k_1},\floor{k_2})\right) < 0.
\end{align*}
The first inequality above uses Lemma~\ref{lemma:ErlangB-LB}, and the
second uses Lemma~\ref{lemma:det_B_bounds}. To see that the second
term is negative, note that
\begin{align*}
v_1u-u_1v &= vu\biggl(\frac{v_1}{v} - \frac{u_1}{u}\biggr)\\
&= vu \left( \left(1 - \frac{N_1 + \ceil{k_2}}{a_1}\right)   -    E(N_1-\ceil{k_1},a_1) \right) \\
&< vu \left( E(N_1 + \ceil{k_2},a_1) - E(N_1-\ceil{k_1},a_1) \right) < 0.
\end{align*}
Both the above inequalities follow from
Lemma~\ref{lemma:ErlangB-LB}. Therefore, we conclude that
$\pd{B_1^{(bo)}(k_1,k_2)}{k_2} < 0.$

Next, we prove that $\pd{B_1^{(bo)}(k_1,k_2)}{k_1} > 0.$ An elementary
calculation yields
\begin{align*}
\pd{B_1^{(bo)}(k_1,k_2)}{k_1} &= \frac{(u_1 d - um_1)+(u_1v-uv_1)\fr{k_2}}{(G')^{2}}
\end{align*}
It suffices to argue that each of the terms in the numerator above is
positive. To see that the first term is positive, note that
\begin{align*}
u_1d-um_{1} &= ud\biggl( \frac{u_1}{u} - \frac{m_{1}}{d}\biggr)\\
&= ud\left(E(N_1-\ceil{k_1},a_1)  - B_1^{(bo)}(\floor{k_1},\floor{k_2}) \right) > 0.
\end{align*}
The inequality above follows from
Lemma~\ref{lemma:det_B_bounds}. Since we have already proved that
$(v_1u-vu_1)<0,$ it follows that the second term is also
positive. This proves that $\pd{B_1^{(bo)}(k_1,k_2)}{k_1} > 0.$

\subsection{Proof of Statement~3}
\label{app:bo_mon_3}
It suffices to prove that $B_{overall}^{(bo)}(k_1,k_2) $ is a strictly
decreasing function of $k_1.$ We first prove the monotonicity over
integer-valued $k_1$ (Lemma~\ref{lemma:sample_path_bo}) and then show
that the monotonicity also extends to real-valued $k_1.$

\begin{lemma}
  \label{lemma:sample_path_bo}
  If $\mu_1 = \mu_2 = \mu,$ then for $k_1 \in \{0,1,\cdots,N_1-1\}$
  and $k_2 \in [0,N_2],$
  $$B_{overall}^{(bo)}(k_1+1,k_2) < B_{overall}^{(bo)}(k_1,k_2).$$
\end{lemma}

\proof[Proof of Lemma \ref{lemma:sample_path_bo}]
  The proof is based on a sample path approach. We assume that call
  holding times for both providers are exponentially distributed with
  parameter $\mu$ (we are free to make this assumption given that the
  blocking probabilities are insensitive to the call holding time
  distributions).

  We consider two systems -- an `old' system (O) with sharing
  configuration $(k_1,k_2)$ and a `new' system (N) with perturbed
  sharing configuration $(k_1+1,k_2).$ In what follows, we will couple
  the arrival processes and service durations across these systems in
  such a way that $N$ system will serve at-least as many calls as the
  $O$ system on any sample path.

  At time zero, we start with both the N and the O system being empty.
  Let $n_i(t)$ denote the number of Provider~$i$ calls in the O system
  at time $t,$ and $\tilde{n}_i(t)$ denote the the number of
  Provider~$i$ calls in the N system at time $t.$ The two systems see
  exactly the same call arrival process.  Moreover, calls that are
  admitted into both systems have the same holding time (this ensures
  that such calls complete at the same time in both systems). Calls
  that are admitted into one of the systems but not into the other are
  categorized as follows.
  \begin{itemize}
  \item Type~1: A Provider~2 call that is admitted into the N system
    but not the O system.
  \item Type~2: A Provider~1 call that is admitted into the N system
    but not the O system.
  \item Type~3: A Provider~2 call that is admitted into the O system
    but not the N system.
  \item Type~4: A Provider~1 call that is admitted into the O system
    but not the N system.
  \end{itemize}
  Note that
  \begin{itemize}
  \item $n_1(t) - \tilde{n}_1(t) = $ $\#$ of Type~4 calls (in O) at
    time $t$ $-$ $\#$ of Type~2 calls (in N) at time $t$
  \item $\tilde{n}_2(t) - n_2(t) = $ $\#$ of Type~1 calls (in N) at
    time $t$ $-$ $\#$ of Type~3 calls (in O) at time $t$
  \end{itemize}

  We will now couple the service durations of Type~$j$ calls in such a
  way that at all times, the states of the N and O systems satisfy one
  of the following three relations:
  \begin{itemize}
  \item {\bf R1}: $n_i(t) = \tilde{n}_i(t)$ for $i = 1,2$
  \item {\bf R2}: $n_1(t) = \tilde{n}_1(t),$ and
    $n_2(t) = \tilde{n}_2(t) - 1$
  \item {\bf R3}: $n_1(t) = \tilde{n}_1(t) + 1,$ and
    $n_2(t) = \tilde{n}_2(t) - 1$
  \end{itemize}
  Note that under all three relations,
  \begin{equation}
    \label{eq:sample_path_bo1}
    \tilde{n}_1(t) + \tilde{n}_2(t) \geq n_1(t) + n_2(t), 
  \end{equation}
  with equality under R1 and R3, and a strict inequality under
  R2. Moreover, 
  \begin{equation}
    \label{eq:sample_path_bo2}
    n_1(t) \geq \tilde{n}_1(t), \quad \tilde{n}_2(t) \geq n_2(t).
  \end{equation}
  The states satisfy R1 at time~0. Also, note that calls that get
  admitted into both systems do not alter the relation between the
  states, at times of arrival or departure. So we only need to focus
  on arrival/departure epochs of Type~j calls, $j \in \{1,2,3,4\}.$
  Our argument will proceed inductively in time.

  \noindent{\bf{Type~1 arrival}:} Suppose that a Type~1 arrival (into
  N) occurs at time $s.$ Since \eqref{eq:sample_path_bo1} holds at
  time $s^-,$ we must have
  \begin{equation}
    \label{eq:sample_path_bo3}
    n_1(s^-) + n_2(s^-) \leq \tilde{n}_1(s^-) + \tilde{n}_2(s^-) < N_1
    + N_2.
  \end{equation}
  This implies that at $s^-,$ the states satisfy R1 with
  $n_2(s^-) = \tilde{n}_2(s^-) = N_2 + k_1.$ In this case, the arrival
  would result in $\tilde{n}_2(s) = N_2 + k_1 + 1,$ implying the
  states would satisfy R2 at time $s.$ The holding time of the newly
  arrived call in N is taken to be an independent $\mathrm{Exp}(\mu)$
  random variable.

  \noindent{\bf{Type~2 arrival}:} Suppose that a Type~2 arrival (into
  N) occurs at time $s.$ This implies \eqref{eq:sample_path_bo3} as
  before.
  It then follows that
  $n_1(s^-) \in \{N_1+\floor{k_2}, N_1+\ceil{k_2}\},$ and
  $\tilde{n}_1(s^-) < n_1(s^-).$ This implies that the states satisfy
  R3 at $s^-,$ and thus satisfy R2 at time $s.$ In other words, at
  time $s,$ we have
  \begin{itemize}
  \item $\#$ of Type~4 calls (in O) $=$ $\#$ of Type~2
    calls (in N)
  \item $\#$ of Type~1 calls (in N) $=$ $\#$ of Type~3
    calls (in O) + 1
  \end{itemize}
  Thus, each Type~4 call can be mapped to a unique Type~2 call, and
  each Type~3 call can be mapped to a unique Type 1 call (one Type~1
  call remains `unmapped'). Given the memorylessness of the call
  holding times, we now re-sample the residual lives of all calls
  using independent $\mathrm{Exp}(\mu)$ random variables such that
  mapped call pairs have the same residual life. This ensures that
  mapped calls (these belong to different systems) depart at the same
  time.

  \noindent{\bf{Type~3 arrival}:} Suppose that a Type-3 arrival (into
  O) occurs at time $s.$ This implies that
  $$n_1(s^-) +n_2(s^-) < \tilde{n}_1(s^-) + \tilde{n}_2(s^-) = N_1 + N_2,$$
  which implies that the states satisfy R2 at $s^-$ and R1 at $s.$
  Thus, at time $s,$ we have
  \begin{itemize}
  \item $\#$ of Type~4 calls (in O) $=$ $\#$ of Type~2
    calls (in N)
  \item $\#$ of Type~1 calls (in N) $=$ $\#$ of Type~3 calls (in O)
  \end{itemize}
  At this point, we map each Type~2 call to a unique Type~4 call, and
  each Type~3 call to a unique Type~1 call. As before, we re-sample
  the residual service times of all calls such that mapped calls have
  the same residual life.

  \noindent{\bf{Type~4 arrival}:} Suppose that a Type-4 arrival (into
  O) occurs at time $s.$ This implies that
  $$n_1(s^-) +n_2(s^-) < \tilde{n}_1(s^-) + \tilde{n}_2(s^-) = N_1 + N_2,$$
  which implies that the states satisfy R2 at $s^-$ and R3 at $s.$
  Thus, at time $s,$ we have
  \begin{itemize}
  \item $\#$ of Type~4 calls (in O) $=$ $\#$ of Type~2
    calls (in N) + 1
  \item $\#$ of Type~1 calls (in N) $=$ $\#$ of Type~3
    calls (in O) + 1
  \end{itemize}
  At this point, we map each Type~2 call to a unique Type~4 call, and
  each Type~3 call to a unique Type~1 call. Finally, we map the
  remaining (yet unmapped) Type~4 call to the remaining (yet unmapped)
  Type~1 call. As before, re-sample the residual life of all calls
  such that mapped calls have the same residual life.

  \noindent{\bf{Departures}:} Based on the above coupling rules for
  residual lives of calls across $O$ and $N,$ we see that four types
  of departures events are possible.

  \begin{itemize}
  \item Simultaneous departure of Type~2 call in N and Type~4 call in
    O: Clearly, the relationship between the states in O and N remains
    unaltered.
  \item Simultaneous departure of Type~1 call in N and Type~3 call in
    O: Clearly, the relationship between the states in O and N remains
    unaltered.
  \item Departure of an unmapped Type~1 call out of N: This can only
    happen if the states satisfy R2 just prior to the departure. The
    states then satisfy R1 post-departure.
  \item Simultaneous departure of a Type~1 call out of N and a Type~4
    call out of O: This can only happen if the states satisfy R3 just
    prior to the departure. Clearly, the states will satisfy R1
    post-departure.
  \end{itemize}

  This completes the argument that the states of the systems O and N
  remain related via R1, R2, or R3 at all times.

  Note that any Type~3/4 departure out of the O system is always
  synchronized with a Type 1/2 departure out of the N system. This
  means that at all times, the cumulative departures out of the N
  system exceed the cumulative departures out of the $O$
  system. Moreover, since there is a positive rate associated with
  `solo' Type 1 departures out the N system, the statement of the
  lemma follows.
\endproof

We are now ready to complete the proof of Statement~3. Given
Lemma~\ref{lemma:sample_path_bo}, it suffices to show that for
$k_1 \in \{0,1,\cdots,N_1-1\}$ and $k_2 \in [0,N_2],$
$B_{overall}^{(bo)}(k,k_2)$ is strictly decreasing over
$k \in [k_1,k_1+1].$ From our gradient calculations, it is not hard to
see that $k \in [k_1,k_1+1],$
\begin{align*}
  \pd{B_{overall}^{(bo)}(k,k_2)}{k} &= \frac{N(k_2)}{(G'(k,k_2))^{2}}.
\end{align*}
Note that the numerator does not depend on $k.$ It thus suffices to
prove that $N(k_2) < 0.$ 

From Lemma~\ref{lemma:sample_path_bo}, invoking the mean value
theorem, it follows that there exists $k' \in (k_1,k_1+1)$ such that
$\frac{N(k_2)}{(G'(k',k_2))^{2}} < 0,$ which implies that
$N(k_2) < 0.$ This completes the proof of Statement~3.

\section{Proof of Lemma~\ref{lemma:pareto_structure} }
\label{app:pareto_struc}

We define a sharing configuration $k = (k_1,k_2)$ to be
\emph{efficient} if there does not exist a sharing configuration $k'$
such that $B_i(k') \leq B_i(k)$ for all $i,$ and $B_i(k') < B_i(k)$
for some $i.$ Let $\mathcal{P}$ denote the set of efficient
configurations. Under this notation, the set $\hat{\mathcal{P}}$ of
Pareto-efficient configurations is given by
\begin{equation*}
  \hat{\mathcal{P}} = \mathcal{P} \cap \mathcal{Q},
\end{equation*}
where $\mathcal{Q}$ denotes the set of QoS-stable configurations. 

The first step of the proof is to show that
$$\mathcal{P} = \mathcal{X} \setminus \mathcal{X}^o.$$ Note that
Lemma~\ref{lemma:pareto_bo_improvement} implies that there are no
efficient sharing configurations in $\mathcal{X}^o.$ Thus, it only
remains to show that any
$\hat{k} \in \mathcal{X} \setminus \mathcal{X}^o$ is efficient. For
the purpose of obtaining a contradiction, suppose that
$\hat{k} \in \mathcal{X} \setminus \mathcal{X}^o$ is not
efficient. Then there exists $\bar{k} \in \mathcal{X}$ such that
$B_i(\bar{k}) \leq B_i(\hat{k})$ for all~$i.$ From the monotonicity of
the blocking probabilities along $\mathcal{X} \setminus \mathcal{X}^o$
(Statements~1 and~2 of Theorem~\ref{thm:monotonicity}), it follows
that $\bar{k} \notin \mathcal{X} \setminus \mathcal{X}^o,$ which
implies that $\bar{k} \in \mathcal{X}^o.$ Now, define for
$i\in \{1,2\},$
$$g_i(k_1,k_2) := \max(0, B_i(\bar{k}_1,\bar{k}_2) - B_i(k_1,k_2)).$$
Consider the following optimization:
$$\max_{k \in \mathcal{X}} g_1(k_1,k_2) g_2(k_1,k_2).$$ Since this is
the maximization of a continuous function over a compact domain, a
maximizer $k^* \in \mathcal{X}$ exists. Moreover, the optimum value is
strictly positive (follows from
Lemma~\ref{lemma:pareto_bo_improvement}),
$k^* \in \mathcal{X} \setminus \mathcal{X}^o,$ and
$B_i(k^*) < B_i(\bar{k})$ for all $i.$ Thus, we have
$\hat{k}, k^* \in \mathcal{X} \setminus \mathcal{X}^o$ such that
$B_i(k^*) < B_i(\hat{k})$ for all $i.$ However, this contradicts the
strict monotonicity of the blocking probabilities over
$\mathcal{X} \setminus \mathcal{X}^o.$ Thus, we conclude that
$\hat{k}$ is efficient.

Having proved that $\mathcal{P} = \mathcal{X} \setminus
\mathcal{X}^o,$ characterizing $\hat{\mathcal{P}}$ boils down to
identifying the subset of QoS-stable sharing configurations
in~$\mathcal{P}.$ For this, consider the three cases in the statement
of the lemma separately. We give the proof for Case~1 here; the proofs
for Cases~2 and~3 are on similar lines and are omitted.

\noindent{\bf Case 1:}  $E(N_1 + N_2,a_1 + a_2) < E(N_i,a_i)$ $\forall\ i$  \\
We have $$B_1(1,1) < E(N_1,a_1) < B_1(1,0).$$ Thus, there a unique
$\hat{k}_2 \in (0,1)$ such that $B_1(1,\hat{k}_2) = E(N_1,a_1)$. It is
easy to see that the set of sharing configurations in~$\mathcal{P}$
where Provider~1 strictly improves upon its standalone blocking
probability is given by
\begin{equation*} \label{eq:pareto_reg_1}
\{(y,1):\ y \in[0,1]\} \cup \{(1,y):\ y \in(\hat{k}_2,1]\}.
\end{equation*} 
Similarly, $$B_2,(1,1) < E(N_2,a_2) < B_2(0,1).$$ Thus, there is a
unique $\hat{k}_1 \in (0,1)$ satisfying $B_2(\hat{k}_1,1) =
E(N_2,a_2)$. As before, the set of sharing configurations
in~$\mathcal{P}$ where Provider~2 strictly improves upon its
standalone blocking probability is given by
\begin{equation*} \label{eq:pareto_reg_2}
\{(y,1):\ y \in(\hat{k}_1,1]\} \cup \{(1,y):\ y \in[0,1]\}.  
\end{equation*}
Thus, the subset of QoS-stable sharing configurations in~$\mathcal{P}$
is the intersection of the above sets.


\section{Proof of Theorem~\ref{thm:qed}}
\label{app_thm:qed}

 We now give the proof of Theorem~\ref{thm:qed}. The key tools in the
 proof are the central limit theorem and Stirling's
 approximation. Since the techniques are somewhat standard, and given
 the space constraints, the proof presentation is terse. The following
 lemma will be used in the proof.
 \begin{lemma}
   \label{lemma:stirling}
   For $X \sim Poisson(a)$ and $\beta \in \R,$ if $N,a \ra \infty$
   such that $\lim_{N \to \infty} (1-\frac{a}{N})\sqrt{N} = \beta$,
   then $\lim_{N\to\infty}\sqrt{N}\prob{X=N} = \phi(\beta).$
   \label{lem:stirling}
 \end{lemma}
 Lemma~\ref{lemma:stirling} follows from an application of Stirling's
 approximation; see \cite{HW81} for a proof.
 
 \proof[Proof of Theorem~\ref{thm:qed}]

  First, we note that for the large system asymptotics, we can ignore
  the fractional part of $k_i,$ and pretend that $k_i$ are
  integers. Indeed, given our monotonicity results for the blocking
  probabilities, it is sufficient to prove the statement of the
  theorem for integral $k_i$ that satisfy \eqref{eqn:sqrt_staff}.

  Our starting point is the expression for the blocking probability in
  Lemma \ref{lemma:bound_sharing}, which we shall rewrite as
  \begin{align*}
    B_i(k_1,k_2) &= \frac{A_i}{G} \\
                 &= \frac{e^{-(a_1+a_2)}A_i}{e^{-(a_1+a_2)}G},
  \end{align*}
  where $G$ is as defined in Lemma \ref{lemma:bound_sharing} and $A_i$
  is the numerator in the expression for $B_i$ in Lemma
  \ref{lemma:bound_sharing}.
  We shall show that
  \begin{equation*}
    \label{eq:qed_pf1}
    \lim_{N\to\infty}e^{-(a_1+a_2)}G = \tilde{G},
  \end{equation*}
  and

  \begin{equation*}
    \label{eq:qed_pf2}
    e^{-(a_1+a_2)}A_i \sim \frac{1}{\sqrt{N}} \tilde{A}_i.
  \end{equation*}
  For this, we define two independent Poisson random variables:
  \begin{align*}
    X_i \sim Poisson(a_i)~,~ i = 1,2
  \end{align*}
  with $a_i$ as in \eqref{eqn:a_lsl}.

  First, we shall prove the limit for $G$. From Lemma
  \ref{lemma:bound_sharing},
  \begin{align}
    G &= 
    \sum_{(n_1,n_2) \in M}
    e^{-(a_1+a_2)} \frac{a_1^{n_1}}{n_1!} \frac{a_2^{n_2}}{n_2!} = \prob{(X_1,X_2) \in M} \nonumber \\
    \label{eq:eq:qed_pf3}
      &= \prob{X_1 \leq N_1 + k_2, X_2 \leq N_2 + k_1,\ X_1 + X_2 \leq N_1 + N_2}.    
  \end{align}
  Now, given the QED scaling under consideration, 
  \begin{align*}
    X_1 \leq N_1 + k_2 &\iff X_1 \leq a_1 - \beta_1\sqrt{a_1} +
                         \gamma_2 \sqrt{\frac{\alpha_2}{\alpha_1}} \sqrt{a_1} +
                         o(\sqrt{a_1}) \\
                       &\iff \frac{X_1 - a_1}{\sqrt{a_1}} \leq -\beta_1 + \gamma_2 \sqrt{\frac{\alpha_2}{\alpha_1}} + \frac{o(\sqrt{a_1})}{\sqrt{a_1}}. 
  \end{align*}
  Similarly, 
  \begin{align*}
    X_2 \leq N_2 + k_1 &\iff \frac{X_2 - a_2}{\sqrt{a_2}} \leq -\beta_2 + \gamma_1 \sqrt{\frac{\alpha_1}{\alpha_2}} + \frac{o(\sqrt{a_2})}{\sqrt{a_2}}.
  \end{align*}
  Finally,
  \begin{align*}
    &X_1 + X_2 \leq N_1 + N_2  \iff \\
    & \quad \sqrt{\frac{a_1}{a_1 + a_2}} \frac{X_1 - a_1}{\sqrt{a_1}}  + \sqrt{\frac{a_2}{a_1 + a_2}} \frac{X_2 - a_2}{\sqrt{a_2}} \leq
      - \frac{\beta_1 \sqrt{a_1} + \beta_2 \sqrt{a_2}}{\sqrt{a_1 + a_2}}
  \end{align*}
  Now, taking limits as $N \ra \infty$ in \eqref{eq:eq:qed_pf3} and
  noting that $\frac{X_1 - a_1}{\sqrt{a_1}}$ and
  $\frac{X_2 - a_2}{\sqrt{a_2}}$ converge in distribution to
  independent standard Gaussians by the central limit theorem, we get
  \begin{align*}
    \lim_{N \ra \infty} G &= P\biggl(Z_1 \leq -\beta_1 + \gamma_2 \sqrt{\frac{\alpha_2}{\alpha_1}},\ Z_2 \leq -\beta_2 + \gamma_1 \sqrt{\frac{\alpha_1}{\alpha_2}}, \\
                          & \quad \sqrt{\alpha_1} Z_1 + \sqrt{\alpha_2} Z_2 \leq -\beta_1 \sqrt{\alpha_1} - \beta_2 \sqrt{\alpha_2} \biggr),
  \end{align*}
  where $Z_1$ and $Z_2$ are iid standard Gaussian random variables. It
  now easy to see that the above equation is equivalent to
  \eqref{eq:geometricG}.

  We now prove the limit of $A_i.$ In terms of $X_1$ and $X_2,$ it is
  easy to see that
  \begin{align*}
    A_i &= \prob{X_i = N_i + k_{-i}}\prob{X_{-i} < N_{-i} -k_{-i}} \\
    &\quad + \prob{X_1 +X_2 = N_1 +N_2, N_1-k_1 \le X_1 \le N_1+k_2} \\
    &=: T_1 + T_2.
  \end{align*}
  We first evaluate the limit of $T_1.$ The limit of the second factor
  of~$T_1$ follows from the central limit theorem. The asymptotic
  behavior of the first factor in $T_1$ can be deduced by invoking
  Lemma \ref{lem:stirling}, whereby we get
\begin{align*}
  \lim_{N\to\infty}\sqrt{N}T_1 = \frac{\phi\left(-\bt{i}+\g{-i}\sqrt{\frac{\alpha_{-i}}{\alpha_{i}}}\right)}{\sqrt{\alpha_i}} \Phi\left(-\bt{-i}-\g{-i}\right).
\end{align*}

Finally, we tackle $T_2.$
\begin{align*}
T_2 =& \sum\limits_{n = -k_1}^{k_2}  e^{-(a_1+a_2)} \frac{a_1^{N_1+n}}{(N_1+n)!} \frac{a_2^{N_2-n}}{(N_2-n)!}\\
=& \sum\limits_{n = -\gamma_1 \sqrt{\alpha_1} \sqrt{N}}^{\gamma_2 \sqrt{\alpha_2} \sqrt{N}}  e^{-(a_1+a_2)} \frac{a_1^{N_1+n}}{(N_1+n)!} \frac{a_2^{N_2-n}}{(N_2-n)!}.
\end{align*}
Again, using Stirling's approximation and rescaling the space by $1/\sqrt{N}$,
 \begin{align*}
	 T_2 &\sim \frac{1}{2\pi\sqrt{N_1N_2}}\sum_{\overset{x = \gamma_2 \sqrt{\alpha_2}}{ x\leftarrow x + 1/\sqrt{N}}}^{\gamma_1 \sqrt{\alpha_1}}
         e^{-\frac{1}{2\alpha_1} (x - \sqrt{\alpha_1}\beta_1)^2} e^{-\frac{1}{2\alpha_2} (-x -  \sqrt{\alpha_2}\beta_2)^2}.
	 \end{align*}
The sum on the right-hand side is a Riemann sum which  when rescaled with $1/\sqrt{N}$ converges to the integral 
\begin{align*}
\int\limits_{ -\gamma_1 \sqrt{\alpha_1}}^{ \gamma_2 \sqrt{\alpha_2}}  \phi\left(\frac{x}{\sqrt{\alpha_1}}-\beta_1\right)
 			\phi\left(-\frac{x}{\sqrt{\alpha_2}}-\beta_2\right) dx,
\end{align*}
and the claimed result follows.
\endproof

\section{Proof of Theorem~\ref{thm:monotonicity} for the probabilistic
  sharing model}
\label{app:monotonicity_prob}

This section is devoted to the proof of Theorem~\ref{thm:monotonicity}
for the probabilistic sharing model.

\subsection{Proof of Statements~1 and~2}

We begin by rewriting the expression for the steady state blocking
probability of Provider~1 (given in Lemma~\ref{lemma:prob_sharing}) as
follows.
\begin{align*}
  B_1^{(p)}(x_1,x_2) = \frac{m_1 +\sum_{i=1}^{N_1}
    u_{1,i}\ x_1^{i}+\sum_{j=1}^{N_2}
    v_{1,j}\ x_2^{j}}{d+\sum_{i=1}^{N_1}
    u_{i}\ x_1^{i}+\sum_{j=1}^{N_2} v_{j}\ x_2^{j} }
\end{align*}
Here, 
\begin{align*}
m_1&= \frac{a_1^{N_1}}{N_1!}\bigl( \sum_{j=0}^{N_2} \frac{a_2^j}{j!}\bigr), \quad u_{1,i} = \frac{a_1^{(N_1-i)}}{(N_1-i)!}\frac{a_2^{(N_2+i)}}{(N_2+i)!},\\
v_{1,j} &=  \left(1 - \frac{N_1
      +j}{a_1} \right)\frac{a_1^{(N_1+j)}}{(N_1+j)!} \bigl( \sum_{i=0}^{N_2-j} \frac{a_2^i}{i!}\bigr), \\
d &= \bigl( \sum_{i=0}^{N_1} \frac{a_1^i}{i!}\bigr) \bigl( \sum_{j=0}^{N_2} \frac{a_2^j}{j!}\bigr),\\
u_{i} &= \frac{a_2^{(N_2+i)}}{(N_2+i)!} \bigl(\sum_{j=0}^{N_1-i} \frac{a_1^j}{j!}\bigr),\quad 
v_{j} = \frac{a_1^{(N_1+j)}}{(N_1+j)!} \bigl( \sum_{i=0}^{N_2-j} \frac{a_2^i}{i!}\bigr).
\end{align*}

To prove that $B_1^{(p)}(x_1,x_2) $ is a strictly decreasing function
of $x_2,$ it suffices to show that $\pd{B_1^{(p)}(x_1,x_2)}{x_2} < 0.$
An elementary calculation shows that
\begin{align*}
 \pd{B_1^{(p)}(x_1,x_2)}{x_2} &= \frac{1}{G^2} \biggl( \sum_{j=1}^{N_2} jx_2^{j-1} [v_{1,j}d -v_{j} m_1] \\
 &\quad + \sum_{j=1}^{N_2} \sum_{i=1}^{N_1} jx_2^{j-1} x_1^i [v_{1,j} u_{i} - v_{j} u_{1,i}] \\
 &\quad + \sum_{j=1}^{N_2} \sum_{i=1}^{N_2} jx_2^{i+j-1} [v_{1,j} v_{i} - v_{j} v_{1,i}] \biggr).
\end{align*}
We now argue that each of the three terms in the above expression is
negative. To see that the first term is negative, note that for any~$j,$
\begin{align*}
  \frac{v_{1,j}}{v_{j}}-\frac{m_1}{d} &= \left(1 - \frac{N_1
      +j}{a_1} \right)- E(N_1,a_1) \\
  & < E(N_1 + j,a_1) - E(N_1,a_1) < 0.
\end{align*}
The inequalities above follow from Lemma~\ref{lemma:ErlangB-LB}. To 
prove that the second term is negative, note that for any $(i,j),$
\begin{align*}
  \frac{v_{1,j}}{v_{j}}-\frac{u_{1,i}}{u_{i}} &= \left(1 -
    \frac{N_1 +j}{a_1} \right)- E(N_1-i,a_1) \\
  &<E(N_1 + j,a_1) - E(N_1-i,a_1) < 0.
\end{align*}
Again, the inequalities above follow from
Lemma~\ref{lemma:ErlangB-LB}. To prove that the third term is
negative, it suffices to show that for $i \neq j,$
\begin{align*}
  &j[v_{1,j} v_{i} - v_{j} v_{1,i}] + i[v_{1,i} v_{j} -
  v_{i} v_{1,j}] \\
  &\quad (j-i)[v_{1,j} v_{i} - v_{j} v_{1,i}] < 0.
\end{align*}
Indeed, 
\begin{align*}
  (j-i)\left[\frac{v_{1,j}}{v_{j}} -
    \frac{v_{1,i}}{v_{i}}\right] &= (j-i)\left[\left(1 - \frac{N_1
        +j}{a_1} \right) - \left(1 - \frac{N_1 +i}{a_1} \right)\right]\\
  &=\frac{-(j-i)^2}{a_1} < 0.
\end{align*}
This proves that $B_1^{(p)}(x_1,x_2) $ is strictly decreasing function
of $x_2.$

To prove that $B_1^{(p)}(x_1,x_2) $ is strictly increasing function
in $x_1,$ we now show that $\pd{B_1^{(p)}(x_1,x_2)}{x_1} > 0.$ An
elementary calculation yields
\begin{align*}
  \pd{B_1^{(p)}(x_1,x_2)}{x_1} &= \frac{1}{G^2} \biggl(
  \sum_{i = 1}^{N_1} ix_1^{i-1} [u_{1,i}d - u_{i} m_1] \\
  &\quad + \sum_{i=1}^{N_1} \sum_{j=1}^{N_2} ix_1^{i-1}x_2^j [u_{1,i}v_{j}
   - u_{i}v_{1,j} ] \\
  &\quad + \sum_{i=1}^{N_1} \sum_{j=1}^{N_1} ix_1^{i+j-1}[u_{1,i}
  u_{j} - u_{i} u_{1,j}] \biggr).
\end{align*}
As before, we argue that each of the terms in the above expression is
positive. To see that the first term is positive, note that for
any~$i,$
\begin{align*}
  \frac{u_{1,i}}{u_{i}}-\frac{m_1}{d} = E(N_1-i,a_1) - E(N_1,a_1)
  > 0.
\end{align*} 
We have already proved that for any $(i,j),$
$v_{1,j} u_{i} - v_{j} u_{1,i} < 0.$ It then follows that the second
term is positive.  Finally, to see that the third term is positive,
note that for any $i \neq j,$
\begin{align*}
  &i[u_{1,i} u_{j} - u_{i} u_{1,j}] + j[u_{1,j} u_{i} -
  u_{j} u_{1,i}] \\
  &\quad = (i-j)[u_{1,i} u_{j} - u_{i} u_{1,j}] \\
  &\quad =(i-j)u_{j}u_{i} \left(\frac{u_{1,i}}{u_{i}} -
    \frac{u_{1,j}}{u_{j}} \right) \\
  &\quad =(i-j)u_{j}u_{i}(E(N_1-i,a_1) - E(N_1-j,a_1)) > 0.
\end{align*}
This proves that $B_1^{(p)}(x_1,x_2) $ is strictly increasing function
in $x_1.$

\subsection{Proof of Statement~3}

We shall prove that the overall blocking probability is decreasing in
the share of each provider. For this, let us generalise the sharing
model by assuming that when there are $n_i$ ongoing calls of type $i$,
an incoming call of type $i$ is accepted with probability
$x_{-i}(n_i)$.
	In the original model,
	\[
		x_{-i}(n_i) = \begin{cases}
					1 & n_i < N_i; \\
					x_{-i} & n_i \geq N_i,
				\end{cases}
	\]  
	in states where $n_1 + n_2 < N_1 + N_2$. 
	
	Let 
	\[
		\sigma(\bn) = \prod_{i=0}^{n_1 -1}\frac{a_1 x_{2}(i)}{i+1}\prod_{j=0}^{n_2 -1}\frac{a_2 x_{1}(j)}{j+1},
	\]
	and
	\begin{align*}
		\sN = \{\bn: n_1 + n_2 < N_1 + N_2\}, \\
		\sNo = \{\bn: n_1 + n_2 \leq N_1 + N_2\}.
	\end{align*}
	Then, the joint stationary probability is
	\[
		\pi(\bm) = \frac{\sigma(\bm)}{\sum_{\bn\in\sNo}\sigma(\bn)},
	\]
	Instead of looking at the blocking probability, we shall look at the probability of accepting a call, which is $1-\bo$. From its definition,
	\begin{align}
	1 - \bo &= \frac{1}{a_1+a_2} \frac{\sum_{\bm\in\sN}(a_1x_2(n_1) + a_2x_1(n_2))\sigma(\bm)}{\sum_{\bm\in\sNo}\sigma(\bm)} \label{eqn:e1} \\
	&= \frac{1}{a_1+a_2} \frac{\sum_{\bm\in\sN}(m_1 + 1)\sigma(\bm + e_1) + (m_2 + 1)\sigma(\bm + e_2)}{\sum_{\bm\in\sNo}\sigma(\bm)} \label{eqn:e2}\\
	&= \frac{1}{a_1+a_2} \frac{\sum_{\bm\in\sNo}(m_1 + m_2)\sigma(\bm)}{\sum_{\bm\in\sNo}\sigma(\bm)} \\
	&=:\frac{1}{a_1+a_2}\frac{A}{D}.
	\label{eqn:pdbov}
	\end{align}
	The local balance equation, $a_ix_{-i}(m_i)\sigma(\bm) = (m_i+1)\sigma(\bm + e_i)$, was used to go from \eqref{eqn:e1} to \eqref{eqn:e2}.
	
	We shall now prove a more general result from which the monotonicity of the blocking probability in the original model shall follow.
	\begin{theorem}
	For $0\leq l < N_1 + N_2$,
	\[
	\pd{\bo}{x_i(l)} < 0, i=1,2.
	\]
\end{theorem}
\begin{proof}
	For a given state $l$ of type $1$, we shall vary the probability $x_2(l)$ while keeping the other probabilities fixed, and show that $1 - \bo$ is increasing in $x_2(l)$. This will show that $B_{ov}$ is decreasing $x_2$ in the original model. A symmetrical argument will hold for $x_1$ as well, which will then complete the proof.
	
	From \eqref{eqn:pdbov},
	\[
	(a_1 + a_2)\pd{(1-\bo)}{x} = \frac{DA^\prime - D^\prime A}{D^2}.
	\]
	
		We shall show that $DA^\prime - D^\prime A > 0$. We have
		\begin{align}
		DA^\prime & =  x^{-1}\left(\sum_{\bm\in\sNo}\sigma(\bm)\right)\left(\sum_{m_1 > l, \bm\in\sNo}(m_1+m_2)\sigma(\bm)\right).				
		\end{align}
		
		Similarly, 
		\begin{align}
		D^\prime A		&= x^{-1}\left(\sum_{m_1 > l, \bm\in\sNo}\sigma(\bm)\right)\left(\sum_{\bm\in\sNo}(m_1+m_2)\sigma(\bm)\right).		
		\end{align}
		Thus, for $x > 0$,
		\begin{equation}
		\frac{DA^\prime}{D^\prime A} = \frac{\bE(M_1 + M_2 \vert M_1 > l)}{\bE(M_1 + M_2)},
		\end{equation}
		where $M_i$ is the number of calls of $P_i$ in stationarity.
		It is thus sufficient to show that the RHS of the above equation is larger than $1$. 
		Let 
		\[
		\gamma(l) := \bE(M_1 + M_2 \vert M_1 \geq l),
		\]
		so that 
		\begin{equation}
		\frac{DA^\prime}{D^\prime A} = \frac{\gamma(l+1)}{\gamma(0)}.	
		\end{equation}
		We shall show that $\gamma(l)$ is increasing in $l$. This is a reasonable assertion because $\gamma(0) = \bE(M_1 + M_2)$ and $\gamma(N_1 + N_2) = N_1 + N_2 \ge \bE(M_1 + M_2)$.
		
		Let $N=N_1 + N_2$. Rewrite $\gamma(l)$ as follows:
		\begin{align}
		\gamma(l) &= \sum_z z \bP(M_1 + M_2 = z\vert M_1\geq l) \\
		&=\frac{\sum_{z=l}^{N}z \sum_{m=0}^z\sigma(z-m,m)}{\sum_{z=l}^{N}\sum_{m=0}^z\sigma(z-m,m)}.
		\end{align}
		Let $c = \left(\sum_{z=l}^{N}\sum_{m=0}^z\sigma(z-m,m)\right)\left(\sum_{z=l+1}^{N}\sum_{m=0}^z\sigma(z-m,m)\right)$. We have \small
		\begin{align}
		c\cdot(\gamma(l+1) - \gamma(l)) &=  \left(\sum_{z=l+1}^{N}z \sum_{m=0}^z\sigma(z-m,m)\right)\left(\sum_{z=l}^{N}\sum_{m=0}^z\sigma(z-m,m)\right) \nonumber\\
		&\qquad - \left(\sum_{z=l}^{N}z \sum_{m=0}^z\sigma(z-m,m)\right)\left(\sum_{z=l+1}^{N}\sum_{m=0}^z\sigma(z-m,m)\right)\\
		&= \left(\sum_{z=l+1}^{N}z \sum_{m=0}^z\sigma(z-m,m)\right)\left(\sum_{m=0}^l\sigma(l-m,m) + \sum_{z=l+1}^{N}\sum_{m=0}^z\sigma(z-m,m)\right)\nonumber\\
		&\qquad - \left(l\sum_{m=0}^l\sigma(l-m,m) + \sum_{z=l+1}^{N}z \sum_{m=0}^z\sigma(z-m,m)\right)\left(\sum_{z=l+1}^{N}\sum_{m=0}^z\sigma(z-m,m)\right)\\
		&=\left(\sum_{z=l+1}^{N}z\sum_{m=0}^z\sigma(z-m,m)\right)\left(\sum_{m=0}^l\sigma(l-m,m)\right)\nonumber\\
		& \qquad - \left(l\sum_{m=0}^l\sigma(l-m,m) \right)\left(\sum_{z=l+1}^{N}\sum_{m=0}^z\sigma(z-m,m)\right)\\
		&>\left(\sum_{z=l+1}^{N}\sum_{m=0}^z\sigma(z-m,m)\right)\left(\sum_{m=0}^l\sigma(l-m,m)\right)\\
		& \geq 0.
		\end{align}\normalsize
		That is,
		$\gamma(l)$ is increasing in $l$.
\end{proof}


\section{Bargaining solutions for logarithmic utilities}
\label{app:logbs}
In this section, we give the definitions and related results for NBS, KSBS and ES for utilities that are a logarithmic function of the blocking probability. 

We do not have a specific result for the logarithmic NBS but we give its definition for the sake of completion. 

\begin{definition}
	A partial sharing configuration $(k_1^*,k_2^*)$ is \emph{ Logarithmic Nash bargaining solution (LOGNBS)}, if the partial sharing configuration satisfies the following condition,
	$$ (k_1^*,k_2^*) = \argmax_{k_i \in [0,1]} \log \left(\frac{B_1(0,0)}{  B_1(k_1,k_2) }\right) \log \left( \frac{B_2(0,0) }{B_2(k_1,k_2)} \right)  . $$
\end{definition}

Similar to NBS, the utility space is not convex. So, LOGNBS may not be unique.

For the logarithmic variants of KSBS and ES the results proved for the linear variants carry over.

\begin{definition}
	A partial sharing configuration $(k_1^*,k_2^*)$ is \emph{ Logarithmic Kalai-Smorodinsky bargaining solution (LOGKSBS)}, if the partial sharing configuration satisfies the following conditions,
	\begin{align*}
	\frac{\log \left( \frac{B_1(0,0)} { B_1(k_1,k_2)}\right)}{ \log \left( \frac{B_2(0,0)} { B_2(k_1,k_2)}\right)} &~= \frac{\log(B_1(0,0)) - \min\limits_{x_i \in [0,1]} \log(B_1(k_1,k_2))}{ \log(B_2(0,0)) - \min\limits_{k_i \in [0,1]} \log(B_2(k_1,k_2))} \\
	B_i(k_1,k_2) &~< B_i(0,0) \quad \forall ~ i
	\end{align*}
\end{definition}

\begin{lemma}
	\label{lemma:LOGKSBS}
	For the bounded overflow sharing model, the LOGKSBS is unique.
\end{lemma}

\begin{definition}
	A partial sharing configuration $(k_1^*,k_2^*)$ is \emph{ Logarithmic egalitarian solution (LOGES)}, if the partial sharing configuration in Pareto set satisfies the following conditions,
	\begin{align*}
	\log(B_1(0,0)) - \log(B_1(k_1,k_2)) =  \log(B_2(0,0)) - \log(B_2(k_1,k_2)) 
	\end{align*}
\end{definition}

Logarithmic egalitarian solution captures the sharing configuration in which both the providers will have the ratio of their blocking probabilities to standalone blocking probabilities to be same.

\begin{lemma}
	\label{lemma:LOGES}
	For the probabilistic sharing model as well as the bounded overflow sharing model, the LOGES is unique. 
\end{lemma}

\begin{corollary}
	\label{cor:loges}
	For uniform standalone blocking probabilities of providers, the LOGES lies at $(N_1,N_2).$
\end{corollary}
Proofs of Lemma \ref{lemma:LOGKSBS} and \ref{lemma:LOGES} use similar arguments to those in the proof of Theorem \ref{lemma:KSBS}.

\end{document}